\documentclass{article}

\usepackage{arxiv}

\usepackage[utf8]{inputenc} 
\usepackage[T1]{fontenc}    
\usepackage{hyperref}       
\usepackage{url}            
\usepackage{booktabs}       
\usepackage{amsfonts}
\usepackage{amsmath}
\usepackage{amssymb}
\usepackage{amsthm}
\usepackage{mathabx}
\usepackage{stmaryrd}
\usepackage{MnSymbol}
\usepackage{nicefrac}       
\usepackage{microtype}      
\usepackage{cleveref}       
\usepackage{lipsum}         
\usepackage{graphicx}
\usepackage{doi}

\newcommand{\defend}{\strut \hfill $\lrcorner$}

\newcommand{\proofsize}{}
\newcommand{\proofend}{}

\newcommand{\showpicture}[1]{#1}
\newcommand{\extakeout}[2]{#1}
\newcommand{\netonesize}{17}
\newcommand{\nettwosize}{11}
\newcommand{\netthreesize}{13}

\newcommand{\trpn}[2]{#1}

\newcommand{\minfigspace}{}

\usepackage[T1]{fontenc}
\usepackage{graphicx}

\usepackage{listings}
\usepackage{booktabs}
\usepackage{multirow}

\usepackage{algorithm}
\usepackage[noend]{algpseudocode}

\usepackage{rotating}
\usepackage{url}

\usepackage{amsfonts}
\usepackage{amsmath}
\usepackage{amssymb}

\usepackage[inline]{enumitem}

\newcommand{\preset}[1]{\ensuremath{{\bullet #1}}}
\newcommand{\postset}[1]{\ensuremath{{#1\bullet}}}
\newcommand{\sources}[1]{\ensuremath{\textit{Src}(#1)}}

\newcommand{\step}[3]{\ensuremath{{#1}\overset{#2}{\rightarrow}{#3}}}
\newcommand{\reachable}[2]{\ensuremath{{#1}{\rightarrow^*}{#2}}}

\newcommand{\act}[2]{\ensuremath{{\overset{*}{#1}_{#2}}}}

\newcommand{\incoming}[2]{\ensuremath{{\triangleright%
\ifx\\#2\\%
\else
_{\scriptscriptstyle{(#2)}}
\fi #1}}}
\newcommand{\outgoing}[2]{\ensuremath{{#1 \triangleleft%
\ifx\\#2\\%
\else
_{\scriptscriptstyle{(#2)}}
\fi}}}
\newcommand{\Loops}[1]{\ensuremath{\textit{Loops}(#1)}}
\newcommand{\Entries}[1]{\ensuremath{\textit{Entries}(#1)}}
\newcommand{\LoopIn}[2]{\ensuremath{\textit{LoopIn}%
\ifx\\#2\\%
\else
_{\scriptscriptstyle{(#2)}}
\fi%
(#1)}}
\newcommand{\Exits}[1]{\ensuremath{\textit{Exits}(#1)}}
\newcommand{\LoopOut}[2]{\ensuremath{\textit{LoopOut}%
\ifx\\#2\\%
\else
_{\scriptscriptstyle{(#2)}}
\fi%
(#1)}}
\newcommand{\Paths}[4]{\ensuremath{\textnormal{\textit{Paths}}%
\ifx\\#3\\%
\else
_{\scriptscriptstyle{(#3)}}
\fi%
(#1,#2)}}

\newcommand{\DoBody}[1]{\ensuremath{\textit{DoBody}(#1)}}

\usepackage{tabularx}
\newcolumntype{L}[1]{>{\raggedright\arraybackslash}p{#1}}
\newcolumntype{C}[1]{>{\centering\arraybackslash}p{#1}}
\newcolumntype{R}[1]{>{\raggedleft\arraybackslash}p{#1}}

\newcommand{\ProofLine}[3]{\noindent \begin{tabular}[t]{R{0.2\textwidth}cL{0.6\textwidth}cC{0.1\textwidth}}%
#1 && #2 && %
\ifx\\#3\\%
\else
 (#3)
\fi%
\end{tabular} \vspace{0.2mm}}

\usepackage{placeins}

\newcommand{\titletext}[0]{Pushing the Limits: Concurrency Detection in \\Acyclic Sound Free-Choice Workflow Nets in $O(P^2 + T^2)$}
\newcommand{\shorttitletext}[0]{Pushing the Limits: Concurrency Detection in $O(P^2 + T^2)$}

\newcommand{\abstracttext}[0]{Concurrency is an important aspect of Petri nets to describe and simulate the behavior of complex systems. 
Knowing which places and transitions could be executed in parallel helps to understand nets and enables analysis techniques and the computation of other properties, such as causality, exclusivity, etc.. 
All techniques based on concurrency detection depend on the efficiency of this detection methodology. 
Kovalyov and Esparza have developed algorithms that compute all concurrent places in $O\big((P+T)TP^2\big)$ for live and bounded nets (where $P$ and $T$ are the numbers of places and transitions) and in $O\big(P(P+T)^2\big)$ for live and bounded free-choice nets. 
Although these algorithms have a reasonably good computational complexity, large numbers of concurrent pairs of nodes may still lead to long computation times. 
This paper complements the palette of concurrency detection algorithms with the \emph{Concurrent Paths} (CP) algorithm for sound free-choice workflow nets. 
The algorithm allows parallelization and has a worst-case computational complexity of $O(P^2 + T^2)$ for acyclic nets and of $O(P^3 + PT^2)$ for cyclic nets. 
Although the computational complexity of cyclic nets has not improved, the evaluation shows the benefits of \emph{CP}, especially, if the net contains many nodes in concurrency relation.}

\newtheorem{theorem}{Theorem}[section]

\newtheorem{corollary}[theorem]{Corollary}
\newtheorem{observation}{Observation}
\newtheorem{definition}{Definition}

\numberwithin{equation}{section}

\title{\titletext}

\usepackage{authblk}

\newbox{\orcid}\sbox{\orcid}{\includegraphics[scale=0.06]{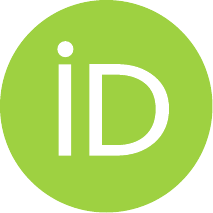}} 
\author[1]{%
	\href{https://orcid.org/0000-0001-9602-3482}{\usebox{\orcid}\hspace{1mm}Thomas M.~Prinz}%
}
\author[2]{%
	\href{https://orcid.org/0000-0002-1498-2653}{\usebox{\orcid}\hspace{1mm}Julien~Klaus}%
}
\author[3]{%
	\href{https://orcid.org/0000-0003-3199-1604}{\usebox{\orcid}\hspace{1mm}Nick R.T.P.~van~Beest}%
}
\affil[1]{Course Evaluation Service, Friedrich Schiller University Jena, Jena, Germany.}
\affil[2]{Faculty of Mathematics and Computer Science, Friedrich Schiller University Jena, Jena, Germany.}
\affil[3]{Data61, Commonwealth Scientific and Industrial Research Organisation (CSIRO), Brisbane, Australia.\newline\texttt{\{Thomas.Prinz,Julien.Klaus\}@uni-jena.de}, \texttt{Nick.VanBeest@data61.csiro.au}}

\renewcommand{\headeright}{Technical Report}
\renewcommand{\undertitle}{Technical Report}
\renewcommand{\shorttitle}{\shorttitletext}

\hypersetup{
pdftitle={\titletext},
pdfsubject={},
pdfauthor={Thomas M.~Prinz, Julien~Klaus, Nick R.T.P.~van Beest},
pdfkeywords={Concurrency detection, Workflow nets, Free-Choice, Soundness},
}

\begin{document}
\maketitle

\begin{abstract}
	\abstracttext
\end{abstract}



\section{Introduction}
\label{sec:Introduction}

Petri nets are a popular and well-studied notion to describe, investigate, revise, and analyze complex system behavior. Especially in the area of information systems and business process management, Petri nets are commonly used to model business processes. Instead of classical control-flow graphs resulting from procedural programming languages, nets are able to model concurrent behavior. Of course, although concurrency is an important aspect of nets (and the business process models they may represent), it complicates their analysis. The detection of concurrent places and transitions of nets helps to understand the behavior of a system. It further enables the measurement of similarity of nets with behavior \cite{DBLP:journals/cii/DijkmanRR12}, the derivation of other behavioral relations, such as e.\,g., causality, exclusivity~\cite{DBLP:conf/apn/PolyvyanyyWCRH14,DBLP:journals/tse/WeidlichMW11}, etc. Some ongoing analysis, therefore, depends on concurrency detection and its efficiency. 

Kovalyov \cite{DBLP:conf/apn/Kovalyov92} has proposed a quintic time algorithm in the number of nodes that computes all concurrent places. Kovalyov and Esparza \cite{KovalyovEsparza} revised that algorithm to have a time complexity of $O\big((P + T)TP^2\big)$ for live and bounded nets and of $O\big(P(P + T)^2\big)$ for live and bounded free-choice nets ($P$ is the set of \emph{places} and $T$ is the set of \emph{transitions} of a net). Both algorithms are well-suited and efficient for their respective problem classes of nets. However, if a net contains many concurrent nodes, the computation time increases significantly during evaluation. 

Decomposition of the net can be used to allow a parallel computation of the concurrency relation and to accelerate the computational method. A decomposition into single-entry single-exit (SESE) fragments helped to speed-up the computation in Weidlich et al. \cite{DBLP:conf/apn/WeidlichPMW10} and Ha and Prinz \cite{DBLP:journals/access/HaP21a}. However, the SESE decomposition approach fails if the net contains inherently unstructured fragments (known as rigids). Weber et al. \cite{DBLP:journals/dpd/WeberHM10} proposed a quadratic time variant to compute concurrent nodes. Their algorithm requires a low number of nodes in the pre- and postset as well as only simple cycles (loops). For other cases, SESE decomposition and the approach of Weber et al. have to utilize the algorithm of Kovalyov and Esparza or more general techniques such as state-space exploration or finite complete prefix unfolding \cite{DBLP:journals/fmsd/EsparzaRV02}. Therefore, these approaches are again at least cubic in time complexity or worse for nets with arbitrary loops.

To overcome these limitations, this paper presents a new algorithm, called the \emph{Concurrent Paths} (CP) algorithm, which is applicable to \emph{sound free-choice workflow nets}. For acyclic nets, it has a quadratic worst-case computation complexity of $O(P^2 + T^2)$ and, for cyclic nets, the worst-case complexity increases to a cubic algorithm of $O(P^3 + PT^2)$. 
The \emph{CP} algorithm is well parallelizable. In addition, the worst-case complexity occurs relatively infrequently, as it depends on the overall number of loops (incl. nested loops) in a net, which is usually small. If this number of loops can be interpreted as a constant, the complexity reduces to be quadratic in average. 
The restriction to connected workflow nets with single start and end places serves to introduce the method, but unconnected nets with several start and end places are also possible. In summary, the \emph{CP} algorithm complements the palette of concurrency detection algorithms for the special case of sound free-choice workflow nets. \trpn{}{A slightly extended version of this paper can be found as a technical report \cite{Prinz2024Preprint}.}

The remainder of the paper is structured as follows. Section~\ref{sec:Preliminaries} explains basic notions, especially nets, paths, loops, markings, and semantics. Subsequently, the concept of concurrency and the algorithm of Kovalyov and Esparza are introduced and described in Section~\ref{sec:Concurrency}. This is followed by revisions of their algorithm to the quadratic time \emph{CP} algorithm for acyclic nets in Section~\ref{sec:AccyclicNets}. Section~\ref{sec:CyclicNets} extends the algorithm to the \emph{Concurrent Paths} algorithm being able to handle cyclic nets. A short outlook in case of inclusive semantics, which is especially used in business process models, is discussed in Section~\ref{sec:InclusiveSemantics}. The evaluation in Section~\ref{sec:Evaluation} demonstrates the strengths and weaknesses of the \emph{CP} algorithm compared to the Kovalyov and Esparza algorithm for live and bounded free-choice nets. Finally, Section~\ref{sec:Conclusion} summarizes the results and provides directions for future work.


\section{Preliminaries}
\label{sec:Preliminaries}

This work is based on well-known definitions of Petri nets. Readers already familiar with these concepts may proceed directly to Section~\ref{subsec:PathToEnd}. However, a quick review is encouraged to ensure alignment with this paper's specific notions.

\subsection{Nets, Paths, Loops, and Workflow Nets}

\begin{definition}[Petri net]
\label{def:PetriNet}
A \emph{Petri net} is a triple $N = (P,T,F)$ with $P$ and $T$ are finite, disjoint sets of \emph{places} and \emph{transitions} and $F \subseteq (P \times T) \cup (T \times P)$ is the \emph{flow} relation. \hfill \defend
\end{definition}

The union $P \cup T$ of a net $N = (P,T,F)$ can be interpreted as \emph{nodes} and $F$ as \emph{edges} between those nodes. For $x \in P \cup T$, $\preset{x} = \{p \; | \; (p,x) \in F \}$ is the \emph{preset} of $x$ (all directly preceding nodes) and $\postset{x} = \{s \; | \; (x,s) \in F\}$ is the \emph{postset} of $x$ (all directly succeeding nodes). Each node in $\preset{x}$ is an \emph{input} of $x$ and each node in $\postset{x}$ is an \emph{output} of $x$. The preset and postset of a set of nodes $X \subseteq P \cup T$ is defined as $\preset{X} = \bigcup_{x \in X} \preset{x}$ and $\postset{X} = \bigcup_{x \in X} \postset{x}$, respectively. 

\begin{definition}[Path]
\label{def:Path}
A \emph{path} $W=(n_1,\ldots,n_m)$ of a net $N = (P,T,F)$ is a sequence of nodes $n_1,\ldots,n_m \in P \cup T$ with $m \geq 1$ and $\forall i \in \{ 1, \ldots, m-1 \}\colon \; n_i \in \preset{n_{i+1}}$. \hfill \defend
\end{definition}

Of course, places and transitions alternate on paths. Nodes are part of a path $A = (x,\ldots,y)$, depicted $x,\ldots,y \in A$. If all nodes of a path are pairwise different, the path is \emph{acyclic}; otherwise, it is \emph{cyclic}. The acyclic path $A$ is sometimes used as the set $\{x,\ldots,y\}$. $\Paths{x}{y}{}{a}$ denotes the set of all \emph{acyclic} paths between nodes $x$ and $y$, where $x,y \in P \cup T$. Two paths $A = (x,\ldots,y)$ and $B = (y,\ldots,z)$ can be \emph{concatenated} to a new path $A + B = (x,\ldots,y) + (y,\ldots,z) = (x,\ldots,y,\ldots,z)$. If $A$ and $B$ are acyclic and $A \cap B =\{ y \}$, then $A + B$ is acyclic since all nodes of $A + B$ are pairwise disjoint. A net $N$ is \emph{cyclic} if there is at least one node $x \in P \cup T$ that has a non-trivial path ($\not = (x)$) to itself.

Cyclic nets contain at least one \emph{loop}. A loop is a subgraph of the net in which each node is reachable from any other node, which can be formally defined as follows:

\begin{definition}[Loop]
\label{def:Loop}
\normalfont
A \emph{loop} $L=(P_L,T_L,F_L)$ of a net $N=(P,T,F)$ is a \emph{strongly connected component} of $N$, i.\,e., $L$ is a maximal subgraph of $N$, such that $P_L \subseteq P$, $T_L \subseteq T$, and $F_L \subseteq F$ \cite{DBLP:books/daglib/0023376}. We denote the set of all loops of $N$ with $\Loops{N}$, such that $\Loops{N} = \emptyset$ for acyclic nets. We further define $\Entries{L}$ as the set $\{ l \in (P_L \cup T_L)\colon \; \preset{l} \not \subseteq (P_L \cup T_L) \}$ of \emph{loop entries} of $L$, and $\Exits{L}$ as the set $\{ l \in (P_L \cup T_L)\colon \; \postset{l} \not \subseteq (P_L \cup T_L) \}$ of \emph{loop exits} of $L$. All flows in $\{(o,l) \in F\colon \; o \notin (P_L \cup T_L) \; \land \; l \in (P_L \cup T_L)\}$ are \emph{loop-entry flows} and all flows in $\{(l,o) \in F\colon \; l \in (P_L \cup T_L) \; \land \; o \notin (P_L \cup T_L)\}$ are \emph{loop-exit flows} of $L$. \hfill \defend
\end{definition}

Each net in this paper is restricted to be \emph{free-choice}: $\forall p \in P\colon \; |\postset{p}| > 1 \; \longrightarrow \; \preset{(\postset{p})} = \{ p \}$. 
Visualized nets have circles representing places, rectangles representing transitions, and directed edges representing flows (see Figure~\ref{fig:ExamplePN}).

\begin{definition}[Workflow and AFW-net]
\label{def:WorkflowNet}
\normalfont
A \emph{workflow net} $WN = (P,T,F,i,o)$ is a net $(P,T,F)$ with $i,o \in P$, $\preset{i} = \emptyset$, and $\postset{o} = \emptyset$. $i$ is the \emph{source} and $o$ is the \emph{sink} of $WN$. All nodes are on a path from $i$ to $o$. If $WN$ is free-choice, we call it \emph{FW-net}. If $WN$ is acyclic and free-choice, we call it \emph{AFW-net}. \hfill \defend
\end{definition}

\noindent The visualized net in Figure~\ref{fig:ExamplePN} is a workflow (FW-)net as well as an AFW-net.

\begin{figure}[tb]
	\centering
		\includegraphics[width=0.50\textwidth]{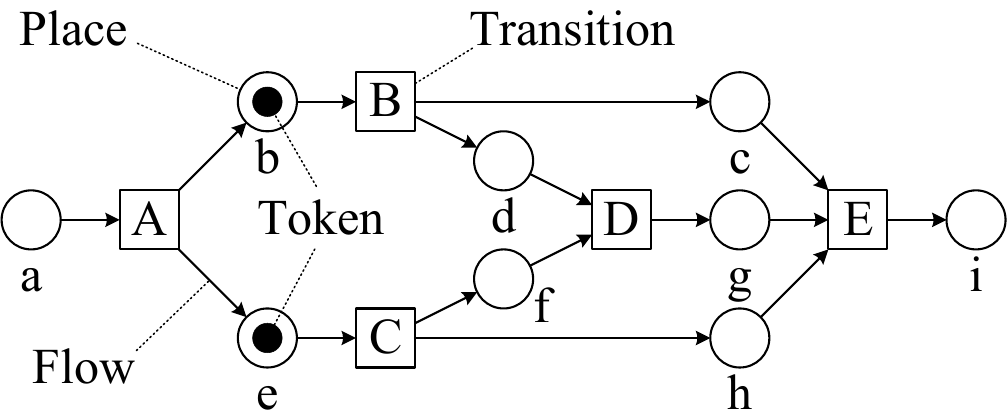}
	\caption{A graphical example of a Petri net.}
	\label{fig:ExamplePN}
\end{figure}

\subsection{Markings, Semantics, Reachability, and Soundness}

The state of a net is described by a so-called \emph{marking}, which specifies the number of \emph{tokens} in each place.

\begin{definition}[Marking]
\label{def:Marking}
A \emph{marking} of a net $N = (P,T,F)$ is a total mapping $M\colon P \mapsto \mathbb{N}_0$ that assigns a natural number (inclusively 0) of \emph{tokens} to each place $P$. $M(p) = 1$ means that place $p \in P$ carries 1 token in marking $M$. The mapping $M$ is sometimes used as the set $\{ p \in P\colon M(p) \geq 1 \}$. \hfill \defend
\end{definition}

An \emph{initial} marking of a net is a special marking with at least one place having a token. The initial marking $M_i$ of a workflow net $WN=(P,T,F,i,o)$ is a marking $\{ i \}$ where only the source $i$ of a net contains exactly one token. A \emph{terminal} marking $M_o$ is a marking $\{ o \}$ of $WN$ where only the sink $o$ contains exactly one token. Transitions whose input places all contain at least one token are \emph{enabled} in a marking and can be fired. This leads to the execution semantics of a net:

\begin{definition}[Execution semantics]
\label{def:Semantics}
Let $N = (P,T,F)$ be a net with a marking $M$. A transition $t \in T$ is \emph{enabled} in $M$ iff every input place of $t$ contains at least one token, $\forall p \in \preset{t}\colon \; M(p) \geq 1$. If $t$ is enabled in $M$, then $t$ can \emph{occur} (``fire''), which leads to a \emph{step} from $M$ to $M'$, denoted $\step{M}{t}{M'}$, with 
\begin{equation}
	M'(p) = M(p) - \begin{cases} 
		1, \; p \in \preset{t} \\ 
		0, \; \text{else} 
	\end{cases} + \begin{cases} 
	  1, \; p \in \postset{t} \\ 
		0, \; \text{else.} 
	\end{cases}
\end{equation}
I.\,e., in a step via $t$, $t$ ``consumes'' one token from all its input places and ``produces'' one token for all its output places. \hfill \defend
\end{definition}

Stepwise firings of transitions lead to chains of fired transitions, which describe the behavior of a net as occurrence sequences:

\begin{definition}[Occurrence Sequences and Reachability]
\label{def:OccurrenceExecutions}
Let $N=(P,T,F)$ be a net with a marking $M_0$. A sequence of transitions $\sigma = \langle t_1,\ldots,t_n \rangle$, $n \in \mathbb{N}_0$, $t_1,\ldots,t_n \in T$, is an \emph{occurrence sequence} of $M_0$ iff $\step{M_0}{t_1}{\step{M_1}{t_2}{\step{\ldots}{t_n}{M_n}}}$ is valid for $M_1,\ldots,M_n$. It can be said that $\sigma$ \emph{leads} from $M_0$ to $M_n$. 

A marking $M'$ is \emph{reachable} from a marking $M$ (denoted $\reachable{M}{M'}$) iff there is an occurrence sequence $\sigma$ of $M$ that leads to $M'$. \hfill \defend
\end{definition}

\noindent Important properties of nets are \emph{liveness} and \emph{boundedness}:

\begin{definition}[Liveness and Boundedness]
A net $N=(P,T,F)$ with its initial marking $M_0$ is \emph{live} iff for every reachable marking $M$, $\reachable{M_0}{M}$, and every $t \in T$, there is a reachable marking $M'$, $\reachable{M}{M'}$, which enables $t$.

$N$ is $n$-\emph{bounded} iff there exists a number $n \in \mathbb{N}_0$ such that for every reachable marking $M$, $\reachable{M_0}{M}$, and for every place $p \in P$ it holds that the number of tokens at $p$ is at most $n$: $\forall p \in P\colon \; M(p) \leq n$. $N$ is \emph{safe} iff it is 1-bounded. \hfill \defend
\end{definition}

\noindent \emph{Soundness} describes an important property of workflow nets:

\begin{definition}[Soundness]
\label{def:Soundness}
\normalfont
A workflow net $WN=(P,T,F,i,o)$ with its initial marking $M_i$, its terminal marking $M_o$, and a reachable state $\reachable{M_i}{M}$ is \emph{sound} iff %
\begin{enumerate}[label=(\arabic*)]
	\item $\reachable{M}{M_o}$,
	\item if $M(o) \geq 1$, then $M = M_o$, and
	\item there is no \emph{dead} transition in $WN$: $\quad \forall t \in T \; \exists M,M'\colon \; \reachable{M_i}{\step{M}{t}{M'}}$. \cite{DBLP:conf/apn/Aalst97} \hfill \defend
\end{enumerate}
\end{definition}

\begin{theorem}
Let $WN=(P,T,F,i,o)$ be a FW-net and $\overline{WN} = \big(P,T \cup \{t'\}, F \cup \{(f,t'), (t',s)\} \big)$ be its ``short-circuit'' net. %
\begin{equation}
	WN \text{ is sound} \quad \Longleftrightarrow \quad \overline{WN} \text{ is live and bounded} \tag*{\defend}
\end{equation}
\end{theorem}

\begin{proof}
See \cite{DBLP:conf/apn/Aalst97}.
\end{proof}

\subsection{Path-to-End Theorem}
\label{subsec:PathToEnd}

To simplify proofs in the rest of this paper, it uses the \emph{Path-to-End Theorem}. It states that in sound FW-nets, there is no reachable marking, in which at least two tokens are on an acyclic path to the sink. 

\begin{theorem}[Path-to-End Theorem]
\label{theorem:PathToEndTheorem}
Let $WN=(P,T,F,i,o)$ be a FW-net with its initial marking $M_i$. %
\begin{gather}
WN \text{ is sound} \nonumber \\ 
\Longrightarrow \\
\forall p \in P \; \forall W \in \Paths{p}{o}{}{a} \; \forall M, \; \reachable{M_i}{M}\colon \quad |M \cap W| \leq 1 \nonumber \tag*{ \defend}
\end{gather}
\end{theorem}

\begin{proof}
The preconditions by Theorem~\ref{theorem:PathToEndTheorem} are a FW-net $WN=(P,T,F,i,o)$ with its initial marking $M_i$. %
The proof is done by showing its contraposition: %
\begin{gather}
\exists p \in P \; \exists W \in \Paths{p}{o}{}{a} \; \exists M, \; \reachable{M_i}{M}\colon \; \; |M \cap W| > 1 \quad \Longrightarrow \quad
WN \text{ is unsound} \label{eq:1CP}
\end{gather}
The proof of the contraposition \eqref{eq:1CP} is done by its contradiction: %
\begin{gather}
\exists p \in P \; \exists W \in \Paths{p}{o}{}{a} \; \exists M, \; \reachable{M_i}{M}\colon \; \; |M \cap W| > 1 \quad \land \quad
WN \text{ is sound} \label{eq:1CD}
\end{gather}
By \eqref{eq:1CD}, let $p \in P$ be a place with an acyclic path $W \in \Paths{p}{o}{}{a}$ to $o$ and $M$ a reachable marking from $M_i$, $\reachable{M_i}{M}$, such that $|M \cap W| > 1$. %
Since $W$ is acyclic, starts in place $p$, and ends in place $o$, it is valid: %
\begin{equation}
	\forall t \in (W \cap T)\colon |\preset{t} \cap W| = 1 \quad \land \quad |\postset{t} \cap W| = 1
\end{equation}
As a consequence, firing a $t \in (W \cap T)$ in any marking $M_1$ cannot reduce the number of tokens on $W$: %
\begin{equation}
	\forall t \in (W \cap T) \; \forall M_1, \; \reachable{M}{M_1}, \; t \text{ is enabled in } M_1, \; \forall M_2, \; \step{M_1}{t}{M_2}\colon \; |M_1 \cap W| = |M_2 \cap W| \label{eq:P2E2}
\end{equation} %
Thus, the only possibility to reduce the number of tokens on $W$ is a place $p_w \in (W \cap P)$ with $|\postset{p_w}| \geq 2$. %
Since $|\postset{o}| = 0$, $p_w \not = o$. %
As a consequence: %
\begin{equation}
	\exists \{ t,t' \} \subseteq \postset{p_w}\colon \; t,t' \in T \quad \land \quad t \notin W \quad \land \quad t' \in W
\end{equation}
Let $t,t'$ be such transitions. %
Since $WN$ is free-choice, it must be valid: %
\begin{equation}
	\preset{t} = \preset{t'} = \{ p_w \} \label{eq:P2E5}
\end{equation} %
It results from \eqref{eq:P2E5}: %
\begin{equation}
	\forall p_w \in (W \cap P), \; |\postset{p_w}| \geq 2 \; \forall t \in \postset{p_w} \; \forall M_1, \; \reachable{M}{M_1}, \; M_1(p_w) \geq 1\colon \; t \text{ is enabled in } M_1 \label{eq:P2E1}
\end{equation} %
We construct an occurrence sequence starting from $M$ with a simple rule: If a $p \in (W \cap P)$ with $|\postset{p}| \geq 2$ has a token in a reachable marking $M_1$, $\reachable{M}{M_1}$, then $t \in (\postset{p} \cap W)$ fires following \eqref{eq:P2E1}. %
We depict a reachable marking under this rule with $\rightarrow^R$ instead of $\reachable{}{}$, i.\,e., $M \rightarrow^R M_1$. %
For this reason, following \eqref{eq:P2E2} and following this rule: %
\begin{equation}
  \forall M_1, \; M \rightarrow^R M_1\colon \; |M_1 \cap W| > 1 \label{eq:P2E3}
\end{equation} %
Since $WN$ is sound by \eqref{eq:1CP} and $W$ is finite, there is at least marking $M_o$, in which no transition is enabled. %
Therefore, let $M_1$ be a reachable marking under this rule, in which no transition is enabled: %
\begin{equation}
	M \rightarrow^R M_1 \; \forall t \in T\colon t \text{ \emph{not} enabled in } M_1 \label{eq:P2E4}
\end{equation} %
Note, it is \emph{not} possible that a transition in \eqref{eq:P2E4} cannot be enabled because of the applied rule. %
Since $WN$ is sound by \eqref{eq:1CD}, the marking $M_1$ by \eqref{eq:P2E4} cannot be the terminal marking $M_o$ by Def.~\ref{def:Soundness} since $|M_1 \cap W| > 1$ by \eqref{eq:P2E3}. %
Therefore, there must be at least one \emph{dead} transition on $W$. %
However, this contradicts soundness by Def.~\ref{def:Soundness}. %
As a consequence, $WN$ cannot be sound. %
This violates contradiction \eqref{eq:1CD}. $\lightning$ %
Therefore, contraposition \eqref{eq:1CP} must be valid. %
This results in the validity of Theorem~\ref{theorem:PathToEndTheorem}. $\checkmark$ \proofend
\end{proof}

\section{Concurrency and the Algorithm of Kovalyov and Esparza}
\label{sec:Concurrency}

A node $x \in P \cup T$ is called \emph{active} in a marking $M$, depicted as $\act{x}{M}$, in the following if it either contains a token ($x \in P$ and $M(x) \geq 1$) or is enabled in $M$ ($x \in T, \forall p \in \preset{x}: M(p) \geq 1$). 

\begin{definition}[Concurrency]
\label{def:Parallel}
For a given workflow net $WN=(P,T,F,i,o)$ and its initial marking $M_i$, there is a \emph{concurrency relation} $\parallel \subseteq (P \cup T) \times (P \cup T)$ with $(x,y) \in\;\parallel$, $x \not = y$, iff there is a reachable marking $M$, $\reachable{M_i}{M}$, with: %
\begin{enumerate}[label=(\arabic*)]
	\item $\act{x}{M}$ and $\act{y}{M}$, and 
	\item if $x \in T$ and $y \in P$, then $\forall z \in \preset{x}\colon \; (z,y)\in\;\parallel$, and
	\item if $x \in P$ and $y \in T$, then $\forall z \in \preset{y}\colon \; (z,x)\in\;\parallel$, and
	\item if $x \in T$ and $y \in T$, then $\forall z \in \preset{x} \; \forall z' \in \preset{y}\colon \; (z,z')\in\;\parallel$.
\end{enumerate} %
$\parallel$ is irreflexive and symmetric in each net (based on thoughts of \cite{KovalyovEsparza}). \hfill \defend
\end{definition}

\noindent The cases (2)-(4) exclude that a place is in concurrency relation with its output transitions and that two transitions are in concurrency relation, which require the same tokens from shared input places. For example, places $e$ and $d$ of Figure~\ref{fig:ExamplePN} and place $d$ and transition $C$ are in a concurrency relation, $\{ (e,d), (d,C) \} \subseteq\;\parallel$. However, $(b,B) \notin\;\parallel$.

Kovalyov and Esparza \cite{KovalyovEsparza} defined a cubic $O\big(|P|(|P| + |T|)^2\big)$ algorithm to identify the $\parallel$ relation for live and bounded free-choice nets. We refer to the algorithm as the \emph{KovEs} algorithm. In the remainder of this section, we provide an overview of its most important concepts and explain its functionality.

\begin{figure}[tb]
	\centering
		\includegraphics[width=0.50\textwidth]{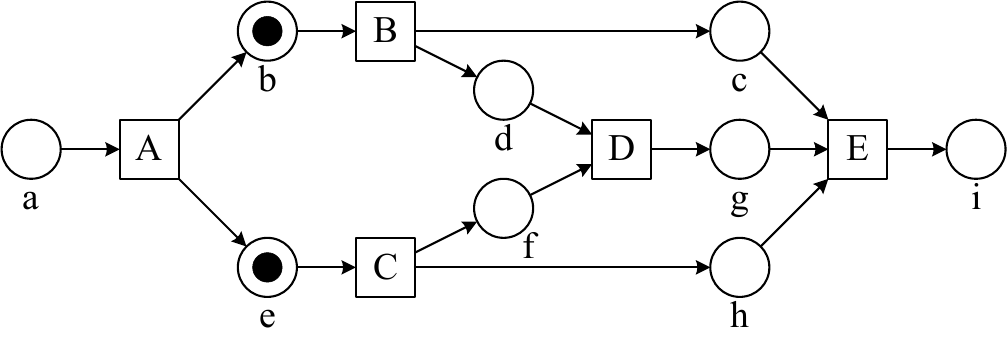}
	\minfigspace
	\caption{A net with concurrency. It holds that $(b,h) \in\;\parallel$ but $(b,g) \notin\;\parallel$.}
	\label{fig:ExamplePNconcpath}
\end{figure}

The initial step of the \emph{KovEs} algorithm is based on a simple observation in live and bounded free-choice nets $N=(P,T,F)$ with an initial marking $M_0$: all places in the postset of transitions are pairwise in a concurrency relation (except for all reflexive pairs): %
\begin{equation}
	\forall t \in T \; \forall x, y \in \postset{t}, \, x \not = y\colon \; (x,y) \in\;\parallel \quad \big( (y,x) \in\;\parallel \big)
\end{equation} %
For example, in Figure~\ref{fig:ExamplePNconcpath}, the pairs $\{ (b,e), (e,b), (c,d), (d,c), (f,h), (h,f) \} \subseteq \parallel$. Besides transitions, all \emph{source places} containing tokens in $M_0$ can be handled in the same way: %
\begin{equation}
	\forall x, y \in P, \; x \not = y, \; M_0(x) \geq 1 \leq M_0(y)\colon \; (x,y) \in\;\parallel \quad \big( (y,x) \in\;\parallel \big)
\end{equation}

The \emph{KovEs} algorithm extends an initial set of the concurrency relation $R$ by considering each already detected pair $(x,y) \in\;\parallel$. Since $(x,y) \in\;\parallel$, $x$ ($y$) may also be concurrent to nodes in $\postset{y}$ ($\postset{x}$). For example, $(b,C)$ is a new candidate for $(b,e)$ in Figure~\ref{fig:ExamplePNconcpath}. For $x$, this means that $\{ (x,s)\colon \; s \in \postset{y} \}$ are new candidate pairs to be concurrent. For each candidate $(x,s)$, there are exactly two cases: %
\begin{enumerate*}[label=\textbf{(\arabic*)}]
	\item $s$ is a place ($s \in P$ and $y \in T$), or 
	\item $s$ is a transition ($s \in T$ and $y \in P$).
\end{enumerate*}
Figure~\ref{fig:ExamplePNconcalgo} visualizes both cases, where node $x$ is visualized as an octagon since it could be either a place or transition. 

\begin{figure}[tb]
	\centering
		\includegraphics[width=0.80\textwidth]{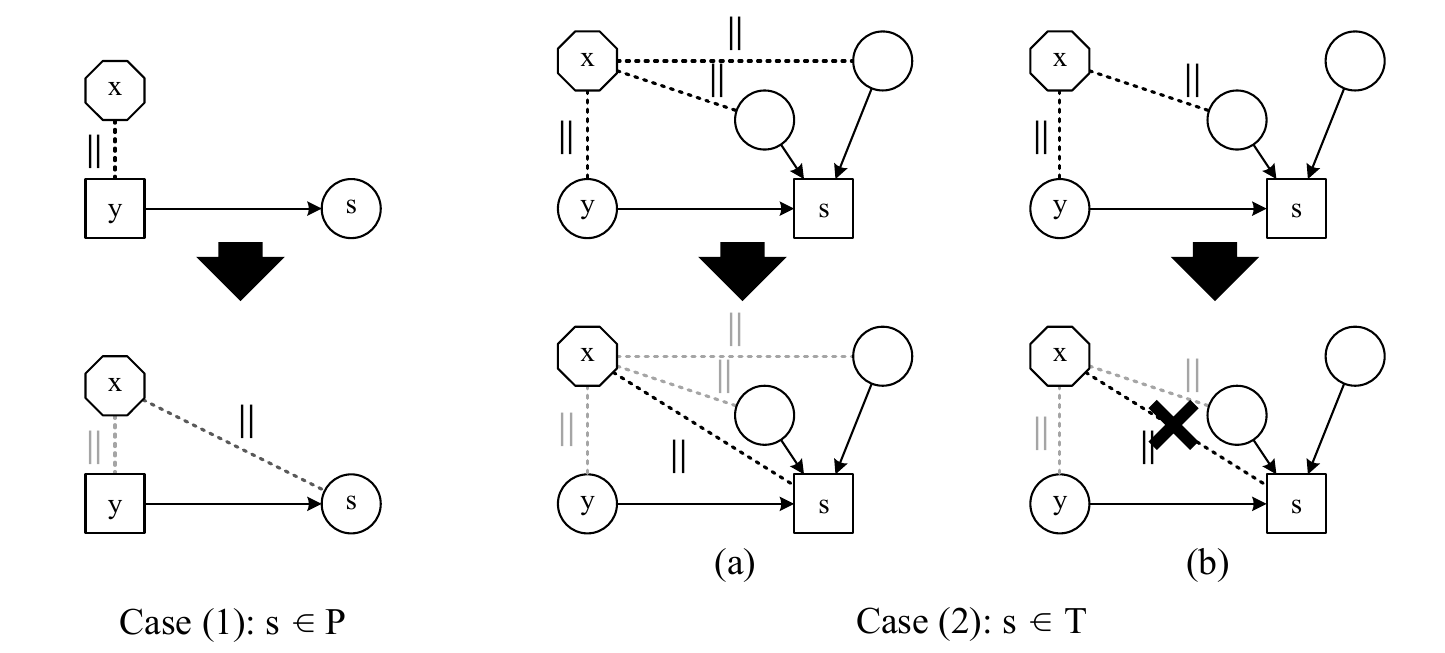}	
	\caption{For a pair $(x,y) \in\;\parallel$, there are two cases of candidate for $(x,s)$, $s \in \postset{y}$: $s \in P$ or $s \in T$. $x$ is visualized as an octagon since it could be either a place or transition. The dotted lines highlight nodes being in a concurrency relation. Arcs denote ordinary flow relations. Above the big black arrows, there are the situations \emph{before} new pairs in relation were derived. Below the black arrows, there are situations with newly derived relations (in black). Note that in Case (2) (b) the relation cannot be derived.}
	\label{fig:ExamplePNconcalgo}
\end{figure}

Consider case \textbf{(1)}: Following from Def.~\ref{def:Parallel} of concurrency, there is a reachable marking $M$ from $M_0$, $\reachable{M_0}{M}$, in which $x$ and $y$ are active, thus, $y$ is enabled. In a step where $y$ fires in $M$, $\step{M}{y}{M'}$, $s$ contains a token in the resulting marking $M'$ and $x$ is still active in $M'$: $x$ and $s$ are active in $M'$. For this reason, for case \textbf{(1)}, all nodes in $\postset{y}$ are concurrent to $x$, leading to the following observation already used in~\cite{KovalyovEsparza}:

\begin{observation}
\label{obs:Case1}
Let $N=(P,T,F)$ be a live and bounded free-choice net. %
\begin{gather}
	\forall x \in (P \cup T) \; \forall  y \in T\colon \nonumber \\
	(x,y) \in\;\parallel \quad \longleftrightarrow \quad \forall s \in \postset{y}\colon (x,s) \in\;\parallel \tag*{\defend}
\end{gather}
\end{observation} %
\noindent The left side of Figure~\ref{fig:ExamplePNconcalgo} visualizes concurrency relations with dotted lines. 

Considering case \textbf{(2)} ($s \in T$ and $y \in P$), there is a reachable marking $M$ from $M_0$, $\reachable{M_0}{M}$, in which $x$ and $y$ are active (i.\,e. $M(y) \geq 1$). Following \cite{KovalyovEsparza}, $(x,s) \in\;\parallel$ if and only if $x$ is concurrent to $\preset{s}$. Otherwise, $s$ and $x$ cannot be active in a reachable marking from $M$. Case \textbf{(2)} (a) of Figure~\ref{fig:ExamplePNconcalgo} illustrates such a situation, in which $x$ is concurrent to $\preset{s}$ and also to $s$, $(x,s) \in\;\parallel$. Case \textbf{(2)} (b) of Figure~\ref{fig:ExamplePNconcalgo} visualizes a situation, in which $x$ is not concurrent to the upper right place but for the other places, such that $(x,s) \notin\;\parallel$. In general, it results in the following observation: 

\begin{observation}
\label{obs:Case2}
Let $N=(P,T,F)$ be a live and bounded free-choice net. %
\begin{gather}
	\forall x \in (P \cup T) \; \forall y \in P \;  \forall s \in \postset{y}\colon \nonumber \\
	\big( \forall p \in \preset{s}\colon \; \; (x,y) \in\;\parallel \; \land \; (x,p) \in\;\parallel \big) \quad \longrightarrow \quad (x,s) \in\;\parallel \tag*{\defend}
\end{gather}
\end{observation}

Algorithm~\ref{algo:KovEsAlgorithm} depicts the resulting algorithm of Kovalyov and Esparza \cite{KovalyovEsparza}. It differs in its notation from the original in four ways: At first, it explicitly ignores reflexive pairs of nodes $(x,x) \in (P \cup T) \times (P \cup T)$; second, it uses the union-operator $\bigcup$ for generating pairs to allow us a simpler discussion about its time complexity; third, line 9 is simplified to just $E \gets R$ instead of $E \gets R \cap \big((P \cup T) \times P\big)$ from the original paper since all pairs in $R$ must fulfill the condition by lines 3 and 5; and, last, it uses the relation $A$ as a mapping in lines 7 and 14 instead of a set.

Readers, who are not familiar with the algorithm, will find an explanation of the algorithm in the following paragraph. Familiar readers may skip it. In Algorithm~\ref{algo:KovEsAlgorithm}, lines 3 and 5 determine the initial concurrency relations discussed as ``initial step'' above. Line 7 defines a mapping $A$ from each place to its directly succeeding places (i.\,e., the union of postsets of the postset of a place). This is used in line 14 to quickly determine new pairs in relation. $E$ in line 9 is a set that contains new pairs not investigated yet. The while loop (lines 10--16) is iterated until $E$ is empty. It takes an arbitrary pair $(x,p)$ out of $E$ (line 11) and one random transition $t$ of the postset of $p$ (line 12). Following case (2) discussed above, $x$ is in concurrency relation with each place in the postset of $t$ (lines 14--16) if it is already in relation with each place $p$ $\preset{t}$ (line 13). $E$ is extended with all new pairs (line 15) and $R$ is extended as well (line 16). Although Kovalyov and Esparza state that their algorithm identifies all pairs of nodes in concurrency relation, it ``only'' identifies pairs of \emph{places} in relation. For this reason, case \textbf{(1)} discussed above is just indirectly important in Algorithm~\ref{algo:KovEsAlgorithm} but will be used for revisions in the next section.

\begin{algorithm}
\caption{The Kovalyov and Esparza algorithm \cite{KovalyovEsparza} to determine $\parallel$ for a given live and bounded free-choice net $N=(P,T,F)$ with its initial marking $M_0$.}
\label{algo:KovEsAlgorithm}
\begin{algorithmic}[1]
\Function{determineKovEsConcurrency}{$N=(P,T,F)$, $M_0$}
	\State \emph{// Add pairs of places carrying a token in $M_0$.}
  \State $R \gets \bigcup_{p_1 \in \sources{N}, M_0(p_1) \geq 1} \, \bigcup_{p_2 \not = p_1 \in \sources{N}, M_0(p_2) \geq 1} \{ (p_1,p_2) \}$
	\State \emph{// Add pairs of places of the postset of each transition.}
	\State $R \gets R \cup \bigcup_{t \in T} \, \bigcup_{p_1 \in \postset{t}} \, \bigcup_{p_2 \not = p_1 \in \postset{t}} \{ (p_1,p_2) \}$
	\State \emph{// Determine for each place $p$ its post-postset.}
	\State $A \gets \bigcup_{p \in P} \, \bigcup_{p' \in \postset{(\postset{p})}} \{ (p,p') \}$
	\State \emph{// Initialize the set of new relations.}
	\State $E \gets R$
	\While{$E \not = \emptyset$}
		\State Take $(x,p) \in E$ and remove it $E \gets E \setminus \{ (x,p) \}$.
		\State Take $t \in \postset{p}$.
		\If{$\bigcup_{y \in \preset{t}} \{ (x,y) \} \subseteq R$}
			\State $tmp \gets \bigcup_{p' \in A(p)} \{ (x,p') \}$
			\State $E \gets E \cup (tmp \setminus R)$
			\State $R \gets R \cup tmp$
		\EndIf
	\EndWhile
	\State \Return $R$
\EndFunction
\end{algorithmic}
\end{algorithm}

\section{Quadratic Algorithm for AFW-Nets}
\label{sec:AccyclicNets}

The \emph{KovEs} algorithm has a worst-case time complexity of $O\big(|P|(|P| + |T|)^2\big)$. This worst case appears if nets have a high level of concurrency, leading to long computation times as the evaluation will show later. This section revises the \emph{KovEs} algorithm for AFW-nets based on soundness and Observations~\ref{obs:Case1} and \ref{obs:Case2} discussed in Section~\ref{sec:Concurrency}. An extension for cyclic FW-nets is presented in the next section. 

Concurrency of two different transitions in sound AFW-nets requires the absence of any path between them:

\begin{theorem}
\label{theorem:PathNotParallel}
Let $WN=(P,T,F,i,o)$ be a sound AFW-net. %
\begin{gather*}
	\forall x,y \in P \cup T, \, x \not = y\colon \nonumber\\
	\big( \Paths{x}{y}{}{a} \cup \Paths{y}{x}{}{a} \big) \not = \emptyset \quad \longrightarrow \quad (x,y) \notin\;\parallel \tag*{\defend}
\end{gather*}
\end{theorem}

\begin{proof}
\proofsize The precondition by Theorem~\ref{theorem:PathNotParallel} is a sound AFW-net $WN=(P,T,F,i,o)$ with its initial marking $M_i$. %
The proof is done by contradiction. %
We assume there is a case, where $x$ has a path to $y$ or $y$ has a path to $x$ and both are still in a concurrency relation: %
\begin{equation}
	\forall x,y \in P \cup T, \, x \not = y\colon \; %
	\big( \Paths{x}{y}{}{a} \cup \Paths{y}{x}{}{a} \big) \not = \emptyset \quad \land \quad (x,y) \in\;\parallel \label{eq:2CD}
\end{equation} %
Without loss of generality and following from this contradiction \eqref{eq:2CD}, let $x$ and $y$ be two nodes with $(x,y) \in\;\parallel$ and there is a path $A_{x \rightarrow y}$ from $x$ to $y$: %
\begin{equation}
	x, y \in P \cup T, \; x \not = y, \; (x,y) \in\;\parallel \quad \land \quad \Paths{x}{y}{}{a} \not = \emptyset \quad \land \quad A_{x \rightarrow y} \in \Paths{x}{y}{}{a} \label{eq:2S1}
\end{equation} %
Following from Def.~\ref{def:Parallel} of concurrency, there is a marking $M$ where $x$ and $y$ are active: %
\begin{equation}
	\exists M, \reachable{M_i}{M}\colon \; \act{x}{M} \; \land \; \act{y}{M} \label{eq:2S2}
\end{equation} %
To simplify ongoing considerations and following from \eqref{eq:2S1} and \eqref{eq:2S2}, we pick two \emph{places} $x'$ and $y'$ carrying a token in $M$. If $x$ ($y$) is a place, then $x' = x$ ($y' = y$); otherwise, $x'$ ($y'$) is a place in $\preset{x}$ ($\preset{y}$): %
\begin{equation}
	x' = \begin{cases} x, \; & x \in P \\ p \in (\preset{x} \cap M), \; & x \in T \end{cases} \quad , \quad %
	y' = \begin{cases} y, \; & y \in P \\ p \in (\preset{y} \cap M), \; & y \in T \end{cases}
\end{equation} %
Since $x'$ and $y'$ are places, $M(x') = M(y') = 1$. Furthermore, let $A_{x' \rightarrow y'}$ be the modified path of $A_{x \rightarrow y}$ with adding $x'$ or removing $y$ if necessary. Following from Def.~\ref{def:WorkflowNet}, there is a path $B_{y' \rightarrow o}$ from $y'$ to the sink $o$ of $WN$: %
\begin{equation}
	\Paths{y'}{o}{}{a} \not = \emptyset \quad \land \quad B_{y' \rightarrow o} \in \Paths{y'}{o}{}{a} \label{eq:2S3}
\end{equation} %
Since $WN$ is acyclic by precondition: %
\begin{equation}
	A_{x' \rightarrow y'} + B_{y' \rightarrow o} = C \quad \land \quad C \text{ is acyclic} \label{eq:2S4}
\end{equation} %
As a consequence from \eqref{eq:2S2}, \eqref{eq:2S3}, and \eqref{eq:2S4}, the sum of all tokens on the acyclic path $C$ in $M$ is more than 2: %
\begin{equation}
	|C \cap M| \geq |\{x',y'\}| \geq 2 \label{eq:2S5}
\end{equation} %
Since $WN$ is sound, \eqref{eq:2S5} contradicts the Path-to-End Theorem~\ref{theorem:PathToEndTheorem} $\lightning$. As a consequence, $WN$ cannot be sound and the contradiction fails. The theorem is valid. $\checkmark$ \proofend
\end{proof}

Following from Theorem~\ref{theorem:PathNotParallel} above, if two nodes are in concurrency relation, there is no path between them.

\begin{corollary}
\label{cor:ParallelPath}
Let $WN=(P,T,F,i,o)$ be a sound AFW-net.
\begin{gather}
	\forall x,y \in P \cup T, \, x \not = y\colon \nonumber \\
	(x,y) \in\;\parallel \quad \longrightarrow \quad \Paths{x}{y}{}{a} = \Paths{y}{x}{}{a} = \emptyset \tag*{\defend}
\end{gather}
\end{corollary}

\begin{proof}
\proofsize The statement directly follows from the contraposition of Theorem~\ref{theorem:PathNotParallel}. \proofend
\end{proof}

For example, the AFW-net in Figure~\ref{fig:ExamplePNconcpath} has concurrent transitions $B$ and $C$, $(B,C) \in\;\parallel$, and concurrent places $b$ and $e$, $(b,e) \in\;\parallel$. After firing transition $C$, $M(f) \geq 1$ and $M(h) \geq 1$, such that $(b,h) \in\;\parallel$. As Cor.~\ref{cor:ParallelPath} states, there is neither a path from place $b$ to $e$ nor from place $b$ to $h$. For place $g$, there is a path from place $b$ to $g$; following Theorem~\ref{theorem:PathNotParallel}, $(b,g) \notin\;\parallel$. Regarding Cor.~\ref{cor:ParallelPath}, each combination of two (different) concurrent \emph{active} nodes cannot have any path in-between for sound AFW-nets.

Although the absence of paths between two nodes is a necessary condition for concurrency in AFW-nets, it is not sufficient. For example, in Figure~\ref{fig:ExamplePNexclpath}, place $c$ has no path to place $d$, but both are never concurrent. 

\begin{figure}[tb]
	\centering
		\includegraphics[width=0.50\textwidth]{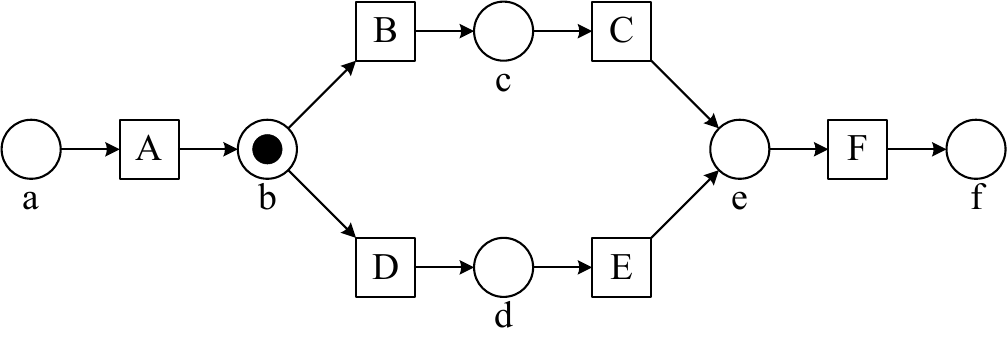}
	\minfigspace
	\caption{A simple net with exclusive paths, in which, for example place $c$ and place $d$ can never be concurrent.}
	\label{fig:ExamplePNexclpath}
\end{figure}

\subsection{First Revision}

The existence or absence of paths between two nodes of an AFW-net is used in our approach to revise the \emph{KovEs }algorithm. In doing this, we reconsider Observations~\ref{obs:Case1} and \ref{obs:Case2} from Section~\ref{sec:Concurrency}. Recall that, given a pair of nodes $(x,y) \in\;\parallel$ and a node $s \in \postset{y}$, we need to determine whether $(x,s) \in\;\parallel$.
Observation~\ref{obs:Case1} with $y \in T$ states that $\forall s \in \postset{y}, (x,s) \in\;\parallel$.
Following  Cor.~\ref{cor:ParallelPath}, there cannot be a path from $x$ to $s$: %

\begin{corollary}
\label{cor:Case1NoPath}
Let $WN=(P,T,F,i,o)$ be a sound AFW-net. %
\begin{gather}
	\forall y \in T \; \forall (x,y) \in\;\parallel \; \forall s \in \postset{y}\colon \nonumber \\
	\Paths{x}{s}{}{a} = \emptyset \quad \longleftrightarrow \quad (x,s) \in\;\parallel \tag*{\defend}
\end{gather}
\end{corollary}

\begin{proof}
\proofsize Constructive proof for $y \in T$, $(x,y) \in\;\parallel$, $s \in \postset{y}$, and the initial marking $M_i$. %
We prove both directions: %
\begin{description}
	\item[$\quad \Paths{x}{s}{}{a} = \emptyset \; \longrightarrow (x,s) \in\;\parallel$] By Cor.~\ref{cor:ParallelPath} and since $(x,y) \in\;\parallel$, there cannot be a path between $x$ and $y$, $\Paths{x}{y}{}{a} = \emptyset$. Therefore, $x \notin \preset{y}$.
	
	Following Def.~\ref{def:Parallel} of concurrency and since $(x,y) \in\;\parallel$, there is a marking $M$ with $\act{x}{M}$ and $\act{y}{M}$: %
	\begin{equation}
		\exists M\colon \; \reachable{M_i}{M} \quad \land \quad \act{x}{M} \quad \land \quad \act{y}{M} \label{eq:C1NP1}
	\end{equation} %
	Since $y$ is active in $M$ and $y \in T$, $y$ is enabled in $M$. As a consequence, firing $y$ in $M$ can lead to a marking $M'$ by Def.~\ref{def:Semantics}, in which $s$ is active: %
	\begin{equation}
		\exists M'\colon \; \step{M}{y}{M'} \quad \land \quad \act{x}{M'} \quad \land \quad \act{s}{M'}
	\end{equation}
	By Def.~\ref{def:Parallel} of concurrency and since $\Paths{x}{s}{}{a} = \emptyset$ and $\Paths{y}{x}{}{a} = \emptyset$ by Cor.~\ref{cor:ParallelPath}: $(x,s) \in\;\parallel$. $\checkmark$
	\item[$\quad (x,s) \in\;\parallel \; \longrightarrow \Paths{x}{s}{}{a} = \emptyset$] By Cor.~\ref{cor:ParallelPath} and since $(x,s) \in\;\parallel$: $\Paths{x}{s}{}{a} = \emptyset$. $\checkmark$ \proofend
\end{description}
\end{proof}

Observation~\ref{obs:Case2} with $y \in P$ states that $x$ is only concurrent to a transition $s \in \postset{y}$ if it is also concurrent to all places $p$ in $\preset{s}$. This statement is complex to check but can be simplified by considering paths: %

\begin{corollary}
\label{cor:Case2NoPath}
Let $WN=(P,T,F,i,o)$ be a sound AFW-net. %
\begin{gather}
	\forall y \in P \; \forall (x,y) \in\;\parallel \; \forall s \in \postset{y}\colon \nonumber \\
	\Paths{x}{s}{}{a} = \emptyset \quad \longleftrightarrow \quad (x,s) \in\;\parallel \tag*{\defend}
\end{gather}
\end{corollary}

\begin{proof}
\proofsize Constructive proof for $(x,y) \in\;\parallel$, $y \in P$, $s \in \postset{y}$, and the initial marking $M_i$. %
Let %
\begin{equation}
	M, \reachable{M_i}{M} \quad \land \quad \act{x}{M} \quad \land \quad \act{y}{M} \label{eq:C6S0}
\end{equation} %
We prove both directions: %
\begin{description}
	\item[$\quad \Paths{x}{s}{}{a} = \emptyset \; \longrightarrow (x,s) \in\;\parallel$] If $x$ would have a path to any node in $\preset{s}$, it would have a path to $s$. However, $\Paths{x}{s}{}{a} = \emptyset$. For this reason, $x$ cannot have a path to any node in $\preset{s}$: %
	\begin{equation}
		\forall p \in \preset{s}\colon \; \Paths{x}{p}{}{a} = \emptyset \label{eq:C6S1}
	\end{equation} %
	As a consequence, since $WN$ is sound, there must be a reachable marking $M'$, in which $x$ is still active and $s$ became enabled (active), i.\,e., firing $s$ is ``independent'' from $x$, because there is no path from $x$ to $s$ by \eqref{eq:C6S1}.: %
	\begin{equation}
		\exists M'\colon \; \reachable{M}{M'} \quad \land \quad \act{x}{M'} \quad \land \quad \act{s}{M'}
	\end{equation} %
	Following Def.~\ref{def:Parallel} and since $\Paths{x}{s}{}{a} = \emptyset$ and $\Paths{y}{x}{}{a} = \emptyset$ by Cor.~\ref{cor:ParallelPath}, $(x,s) \in\;\parallel$. $\checkmark$ %
	\item[$\quad (x,s) \in\;\parallel \; \longrightarrow  \Paths{x}{s}{}{a} = \emptyset$] By Cor.~\ref{cor:ParallelPath} and since $(x,s) \in\;\parallel$: $\Paths{x}{s}{}{a} = \emptyset$. $\checkmark$ \proofend
\end{description}
\end{proof}

In summary, for both cases $y \in P$ and $y \in T$, if and only if $x$ has no path to $s$, then $(x,s) \in\;\parallel$. Otherwise, $(x,s) \notin\;\parallel$. This is summarized by the following theorem: %

\begin{theorem}
	\label{theorem:PostSetNoPathParallel}
	Let $WN=(P,T,F,i,o)$ be a sound AFW-net. %
	\begin{gather}
		\forall (x,y) \in\;\parallel \; \forall s \in \postset{y}\colon \nonumber \\
		\Paths{x}{s}{}{a} = \emptyset \quad \longleftrightarrow \quad (x,s) \in\;\parallel \tag*{\defend}
	\end{gather}
\end{theorem} 

\begin{proof}
\proofsize The theorem combines all cases from Cor.~\ref{cor:Case1NoPath} and Cor.~\ref{cor:Case2NoPath}. $\checkmark$ \proofend
\end{proof}

This fact can already be used to revise the \emph{KovEs} algorithm. However, we consider another revision.

\subsection{Second Revision}

Following the first revision, paths play a crucial role to identify concurrency. For this reason, we define an auxiliary relation: %

\begin{definition}[$HasPath$ relation]
	\label{def:HasPath}
	Let $WN=(P,T,F,i,o)$ be an AFW-net. The $HasPath$ relation %
	\begin{equation}
		HasPath = \big \{ (x,y) \in (P \cup T) \times (P \cup T) \; | \; \Paths{x}{y}{}{a} \not = \emptyset \big \}
	\end{equation} %
	specifies whether a node $x$ has an acyclic path to node $y$ ($HasPath$ is reflexive, so $(x,x) \in HasPath$). $HasPath(x)$ denotes the set of all nodes to which $x$ has a path (again, inclusive of itself). \hfill \defend
\end{definition}

In the following, let $(x,y) \in\;\parallel$ for a sound AFW-net. From $(x,y) \in\;\parallel$ and Cor.~\ref{cor:ParallelPath} it follows that $y \notin HasPath(x)$ and $x \notin HasPath(y)$. Let us consider any node $a \in HasPath(x)$ to which $x$ \emph{has} a path. If $a$ would have a path to $y$ (i.\,e., $y \in HasPath(a)$), then $x$ would have a path to $y$ via $a$. For this reason, no node $a \in HasPath(x)$ has a path to $y$ and no node in $HasPath(y)$ has a path to $x$.

Regarding $(x,y) \in\;\parallel$, there are exactly two cases: \textbf{(a)} No path starting in $x$ to any sink crosses any path starting in $y$ to any sink (i.\,e., $HasPath(x) \cap HasPath(y) = \emptyset$, cf. Figure~\ref{fig:HelpPNHasPath} Case (a)); or \textbf{(b)} at least one path starting in $x$ to a sink crosses at least one path starting in $y$ to a sink (i.\,e., $HasPath(x) \cap HasPath(y) \not = \emptyset$, cf. Figure~\ref{fig:HelpPNHasPath} Case (b)). 

To case \textbf{(a)}: each node in $HasPath(x)$ must be in concurrency relation with any node in $HasPath(y)$ after a stepwise application of Theorem~\ref{theorem:PostSetNoPathParallel} regarding the first revision (and, since tokens in both sets can never converge). This is illustrated in Figure~\ref{fig:HelpPNHasPath} by nodes $a$ and $b$ being concurrent. 

To case \textbf{(b)}: the nodes of $HasPath(x)$ are partially overlapping with the nodes of $HasPath(y)$. This is illustrated in Figure~\ref{fig:HelpPNHasPath} Case (b) with a gray triangle subset. Let $R_{x}^{\bar{y}} = HasPath(x) \setminus HasPath(y)$ be the subset of nodes of $HasPath(x)$ to which $y$ has no path (and, therefore, no node in $HasPath(y)$ as explained previously). A node $a \in R_{x}^{\bar{y}}$ has a subset $HasPath(a) \subseteq HasPath(x)$ (illustrated as the grid triangle subset in Figure~\ref{fig:HelpPNHasPath}). Although no node in $HasPath(y)$ has a path to $a$, $a$ may have paths to nodes in $HasPath(y)$ (more precisely, to nodes in $HasPath(a) \cap HasPath(y)$; in Figure~\ref{fig:HelpPNHasPath} this is the subset of the gray triangle intersecting the grid triangle). Regarding Theorem~\ref{theorem:PathNotParallel}, $a$ cannot be concurrent to those nodes ($c$ in the illustration). $R_{y}^{\bar{a}} = HasPath(y) \setminus HasPath(a)$ contains all nodes of $HasPath(y)$ to which $a$ has no path (and so no other node of $HasPath(a)$). As a consequence, for each node $a \in R_{x}^{\bar{y}}$, there are sets $HasPath(a)$ and $HasPath(y) \setminus HasPath(a) = R_{y}^{\bar{a}}$ being disjoint. There are disjoint paths from $x$ to each node in $HasPath(a)$ and from $y$ to each node in $R_{y}^{\bar{a}}$. Following the first revision and Theorem~\ref{theorem:PostSetNoPathParallel}, a step-wise consideration of these paths  leads to concurrency between all nodes in $HasPath(a)$ and $R_{y}^{\bar{a}}$. Instead of considering these nodes step-by-step, their pairwise concurrency can be added directly, i.\,e., $HasPath(a) \times R_{y}^{\bar{a}} \; \subseteq \; \parallel$. This leads to a revised algorithm for sound AFW-nets with a quadratic computational complexity.

\begin{figure}[tb]
	\centering
		\includegraphics[width=0.50\textwidth]{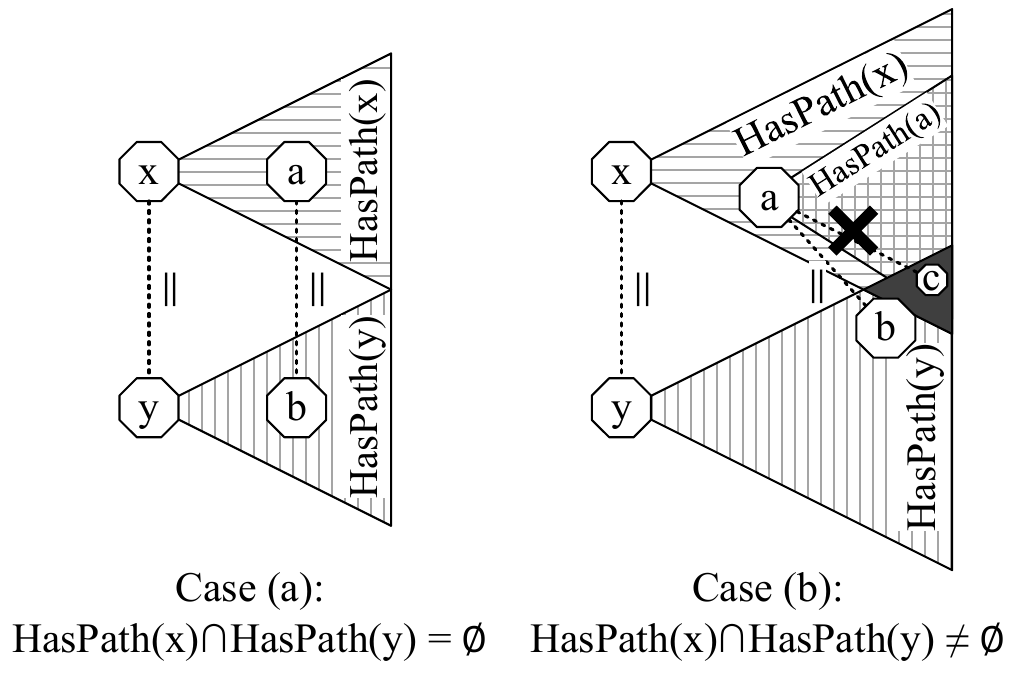}
	\caption{Two cases for mappings $HasPath(x)$ and $HasPath(y)$ of a concurrent pair $(x,y) \in\;\parallel$.}
	\label{fig:HelpPNHasPath}
\end{figure}

\subsection{Revised Algorithm}

Algorithm~\ref{alg:Concurrency} defines the revised algorithm for sound AFW-nets. The algorithm computes the concurrency relation for each node of the net. However, if necessary, it can be modified to just compute the relations for places. Furthermore, the algorithm assumes the relation $\parallel$ to be represented as an \emph{adjacency list}. Naturally, it takes just $O(|\parallel|)$ to put it into a set of pairs. 

Lines~2--6 of Algorithm~\ref{alg:Concurrency} initialize the algorithm. Lines~3--5 initialize the set $R$ for each unprocessed node with an empty set (following the \emph{KovEs} algorithm, which is here the adjacency list of concurrent nodes). This initialization is linear to the number of nodes, $O(|P| + |T|)$. Line~6 computes the $HasPath$ relation as adjacency list by calling a function $computeHasPath$ of Algorithm~\ref{alg:HasPath}. The algorithm computes $HasPath$ utilizing a reverse topological order (each node appears \emph{after} all its output nodes). Such order can be computed in $O(|P| + |T| + |F|)$ \cite{DBLP:books/daglib/0023376}. Thus, Algorithm~\ref{alg:HasPath} can compute $HasPath$ for all nodes in $O(|P| + |T| + |F|)$ as well since each node ($|P \cup T|$) is investigated with its outputs ($|F|$). Assuming that set operations (such as $\cup$) can be achieved in constant time, e.\,g., with a \emph{BitSet}, $HasPath$ can be computed for all nodes in $O(|P| + |T| + |F|)$. 

The actual computation of the concurrency relation takes place in lines~8--17. It creates a new adjacency list $I$ for initial places in concurrency relation (line~8). It investigates all transitions $t$ (line~9) and its output places $x$ (line~10). Each $x$ is initial concurrent to all other output places of $t$ (line~11). Lines~12--13 investigate all pairs of different places being initially in concurrency relation in $I$. Line~14 computes the set $R_{x}^{\bar{y}}$ containing all nodes to which $x$ but not $y$ has a path. As discussed in the \emph{Second Revision}, each node $a \in R_{x}^{\bar{y}}$ is investigated in lines~15--16, where line~16 adds all nodes in $HasPath(y) \setminus HasPath(a)$ to be concurrent to $a$ (i.\,e., all nodes in $HasPath(y)$ to which $a$ does not have a path; recall that $a$ does not have a path to any node of $HasPath(y)$ because it is in $R_{x}^{\bar{y}}$).

\begin{algorithm}[tb]
\caption{The acyclic \emph{Concurrent Paths} algorithm to determine $\parallel$ for a given sound AFW-net $WN=(P,T,F,i,o)$.}
\label{alg:Concurrency}
\begin{algorithmic}[1]
\Function{determineConcurrency}{$WN=(P,T,F,i,o)$, $R = \emptyset$}
	\State \emph{// Initialize}
	\ForAll{$x \in P \cup T$}
		\If{$x \notin R$}
			\State $R(x) \gets \emptyset$
		\EndIf
	\EndFor
	\State $HasPath \gets$ \Call{computeHasPath}{$WN$}
	\State \emph{// Compute}
	\State $I \gets \emptyset$
	\ForAll{$t \in T$}
		\ForAll{$x \in \postset{t}$}
			\State $I(x) \gets I(x) \cup \big(\postset{t} \setminus \{ x \}\big)$
		\EndFor
	\EndFor
	\ForAll{keys $x \in I$}
		\ForAll{$y \in I(x)$}
			\State $R_{x}^{\bar{y}} \gets HasPath(x) \setminus HasPath(y)$
			\ForAll{$a \in R_{x}^{\bar{y}}$}
				\State $R(a) \gets R(a) \cup \big( HasPath(y) \setminus HasPath(a) \big)$
			\EndFor
		\EndFor
	\EndFor
	\State \Return $R$
\EndFunction
\end{algorithmic}
\end{algorithm}

\begin{algorithm}[tb]
\caption{Computation of $HasPath$ as adjacency list for a given AFW-net $WN=(P,T,F,i,o)$.}
\label{alg:HasPath}
\begin{algorithmic}[1]
\Function{computeHasPath}{$WN=(P,T,F,i,o)$}
	\State \emph{// Initialize}
	\State $HasPath \gets \emptyset$
	\ForAll{$x \in P \cup T$}
		\State $HasPath(x) \gets \emptyset$
	\EndFor
	\State $L \gets P \cup T$ in reverse topological order starting from $o$
	\ForAll{$x \in L$}
		\State $HasPath(x) \gets \{ x \} \cup \bigcup_{s \in \postset{x}} HasPath(s)$
	\EndFor
	\State \Return $HasPath$
\EndFunction
\end{algorithmic}
\end{algorithm}

The time complexity for lines~8--16 seems biquadratic, $O(X^4)$, at the first view. However, let us change the perspective on the algorithm focusing on line~16. If line~16 \emph{adds} new information to $R(a)$, at least one new pair in concurrency was detected. That means if line~16 \emph{always} adds new information $R(a)$ and, therefore, finds a new pair in concurrency, lines~8--16 can execute line~16 \emph{at most} $|\parallel| \leq |(P \cup T) \times (P \cup T)|$ times, i.\,e., $O(|P|^2 + |T|^2)$ --- the algorithm would be quadratic. In other words, it should not happen that line~16 is executed without necessity. This is actually the case as we will explain in the following. 

The main task of lines~8--13 is to create and consider a \emph{new} pair $(x,y)$, which is taken from each $\postset{t}$ of $t$ in lines~8--11. There can only be $|P|^2$ of such pairs (as adjacency list) in $I$. Apart from this pair (and, of course, $(y,x)$), let us consider all other ``initial'' pairs $(\alpha,\beta) \not = (x,y) \not = (y,x)$ of lines~8--11. If for transition $t$ with $x,y \in \postset{t}$ it is valid that $t \in HasPath(\alpha)$, then $\postset{t} \subset HasPath(\alpha)$, and if it is valid that $t \in HasPath(\beta)$, then $\postset{t} \subset HasPath(\beta)$. For this reason, it is valid that either: %
\begin{enumerate}
	\item $\{ x,y \} \subset HasPath(\alpha)$, 
	\item $\{ x,y \} \subset HasPath(\beta)$, or 
	\item $\{ x,y \} \not \subset \big( HasPath(\alpha) \cup HasPath(\beta) \big)$.
\end{enumerate} %
As a consequence, all ``initial'' pairs $(x,y)$ and $(y,x)$ cannot be added by the handling of other ``initial'' pairs. The same holds true for all nodes in $R_{x}^{\bar{y}}$ of line~14 whose information to be concurrent to $y$ is firstly added in line~16. Line~16, therefore, \emph{always} adds new information. As explained earlier, line~16 can at most be executed in quadratic complexity making lines~8--16 quadratic, $O(|P|^2 + |T|^2)$. 

In summary, assuming that set operations are applicable in constant time (like $\cup$ and $\setminus$), e.\,g., by a \emph{BitSet}, Algorithm~\ref{alg:Concurrency} can be computed in $O\big((|P| + |T|) + (|P| + |T| + |F|) + |P| + |P|^2 + |T|^2\big) = O\big(3|P| + 2|T| + |F| + |P|^2 + |T|^2\big) = O(|P|^2 + |T|^2)$ in the worst-case (as $|F|$ is quadratic to $|P| + |T|$ in the worst-case). Since at least those pairs in concurrency must be investigated, which are a subset of $(P \cup T) \times (P \cup T)$, it will be difficult to obtain a faster algorithm in terms of asymptotic time complexity than $O(|\parallel|) \subseteq O(|P|^2 + |T|^2)$.

\section{Cubic Algorithm for FW-Nets}
\label{sec:CyclicNets}

Although AFW-nets are not unusual, concurrency in cyclic sound FW-nets must be considered as well. In contrast to the \emph{KovEs} algorithm, Algorithm~\ref{alg:Concurrency} cannot be applied to cyclic nets since the revisions are applicable to acyclic nets only. To overcome this situation, we use a method called \emph{loop decomposition} \cite{DBLP:conf/bpm/PrinzCH22} to decompose a cyclic sound FW-net into a set of sound AFW-nets with the same behavior. Cyclic nets become acyclic by replacing loops with single \emph{loop places}. Figure~\ref{fig:ExampleCyclicPN} shows a cyclic net on the left side, which is decomposed into three acyclic nets on the right side. Subsequently, we apply Algorithm~\ref{alg:Concurrency} to each acyclic net and combine all collected concurrency information, i.\,e., all nodes concurrent to an inserted \emph{loop place} are concurrent to all nodes in the corresponding loop.

\subsection{Loop Decomposition}

The method of loop decomposition \cite{DBLP:conf/bpm/PrinzCH22} was introduced to decompose sound \emph{workflow graphs} with loops into sets of sound workflow graphs without loops. This decomposition method was slightly revised in~\cite{Prinz2023Preprint}. Favre et al. \cite{DBLP:journals/is/FavreFV15} have shown how free-choice nets can be transferred into workflow graphs and vice versa. The notion of soundness used in the work of loop decomposition is equal to soundness of FW-nets \cite{DBLP:journals/dke/FahlandFKLVW11}. For this reason, loop decomposition can be applied to sound FW-nets without strong modifications. Algorithm~\ref{algo:LoopDecomposition} describes the algorithm abstractly. The following only focuses on the consequences for concurrency detection. For further details about loop decomposition, we refer to previous work~\cite{DBLP:conf/bpm/PrinzCH22,Prinz2023Preprint}.

\begin{algorithm}[tb]
\caption{Loop decomposition of a given sound FW-net $WN=(P,T,F,i,o)$.}
\label{algo:LoopDecomposition}
\begin{algorithmic}[1]
\State $Connections \gets \emptyset$
\State $AcyclicNets \gets \emptyset$
\Function{decomposeLoops}{$WN=(P,T,F,i,o)$}
	\State Identify $\Loops{WN}$.
	\If{$|\Loops{WN}| = 0$}
		\State $AcyclicNets \gets AcyclicNets \cup \{ WN \}$
		\State \Return $AcyclicNets$, $Connections$
	\EndIf
	\ForAll{$L=(P_L,T_L,F_L) \in \Loops{WN}$}
		\State Identify loop entries $\Entries{L}$ and loop exits $\Exits{L}$.
		\State Identify do-body $\DoBody{L}$ of $L$.
		\State Copy do-body and relink flows.
		\State Replace $L$ with place $p_L$ in $WN$ and relink flows to $p_L$.
		\State $Connections \gets Connections \cup \{ \big(p_L, P_L \cup T_L\big) \}$
		\State Split $L$ into \emph{loop fragments}.
		\State Insert one source and one sink place to combine loop fragments.
		\State Create new net $WN_L$ of it.
		\State \Call{decomposeLoops}{$WN_L$}
	\EndFor
	\State \Call{decomposeLoops}{$WN$}
	\State \Return $AcyclicNets$, $Connections$
\EndFunction
\end{algorithmic}
\end{algorithm}

Loop decomposition identifies loops as strongly connected components (cf. Def.~\ref{def:Loop} of loops). Following \cite{DBLP:conf/bpm/PrinzCH22}, entries and exits of loops are places in sound FW-nets. The (sometimes unconnected) subgraph between all loop entries and those loop exits being reachable without passing another loop exit is called the \emph{do-body} of the loop. Once a loop exit contains a token, no other place in the loop or the do-body contains a token \cite{DBLP:conf/bpm/PrinzCH22}. Therefore, the do-body can be interpreted as an implicit, initial ``converging area'' of different concurrent tokens before first loop exits are reached. Loop decomposition duplicates the do-body as an explicit, initial converging area before the loop. Subsequently, it replaces the entire loop (without the copied do-body) with a single \emph{loop place} representing the previous loop. All flows into and out of the loop are redirected to start from and end at the loop place. Repeating this procedure with any (nested) loop finally leads to an acyclic sound FW-net \cite{DBLP:conf/bpm/PrinzCH22}.

The extracted loops with all their nodes and flows are decomposed by removing all incoming flows of loop exits. As a consequence, the loop disintegrates into at least one \emph{loop fragment}. In the original method of loop decomposition \cite{DBLP:conf/bpm/PrinzCH22}, each fragment is extended to an own net. This is unnecessary in this context of sound FW-nets, since all fragments of the same loop are mutually exclusive \cite{DBLP:conf/bpm/PrinzCH22}, i.\,e., no node of one fragment can ever be in a concurrency relation with a node of another fragment of the same loop. Therefore, all fragments of a loop together get a new single source place and a new single sink place. For each loop entry and exit, a new transition is inserted connecting the source place with the entry/exit. For each loop exit, a new transition is inserted to connect it with the new sink place. In addition, all transitions previously in the preset of a loop exit are connected to the new sink place. The resulting ``loop'' net is sound and free-choice. In case that this \emph{loop net} is still cyclic (in case of nested loops), loop decomposition can be recursively applied to this net again. This procedure finally terminates in only sound AFW-nets \cite{DBLP:conf/bpm/PrinzCH22}. 

Figure~\ref{fig:ExampleCyclicPN} illustrates the decomposition of a cyclic FW-net into its set of AFW-nets. The cyclic net on the left contains two loops: One with place $b$ as entry and $f$ as exit and one with place $h$ as entry and $j$ and $k$ as exits. The net is decomposed into three acyclic nets (right). The net at the top is the loop-reduced version of the cyclic net. The do-bodies (loop 1: $\{ b, B, c, d, C, e, D \}$, loop 2: $\{ h, H, i, I \}$) remained and the loops were replaced with loop places ($\alpha$ and $\beta$). The acyclic net in the middle corresponds to the upper loop place $\alpha$ and spans from before its loop exit to its loop exit. The net at the bottom corresponds to the lower loop place $\beta$. It consists of two fragments (from exit $j$ to $k$ and from exit $k$ to $j$) and is constructed by adding a source place, sink place, and the remaining transitions.

\begin{figure}[tb]
	\centering
		\includegraphics[width=1.00\textwidth]{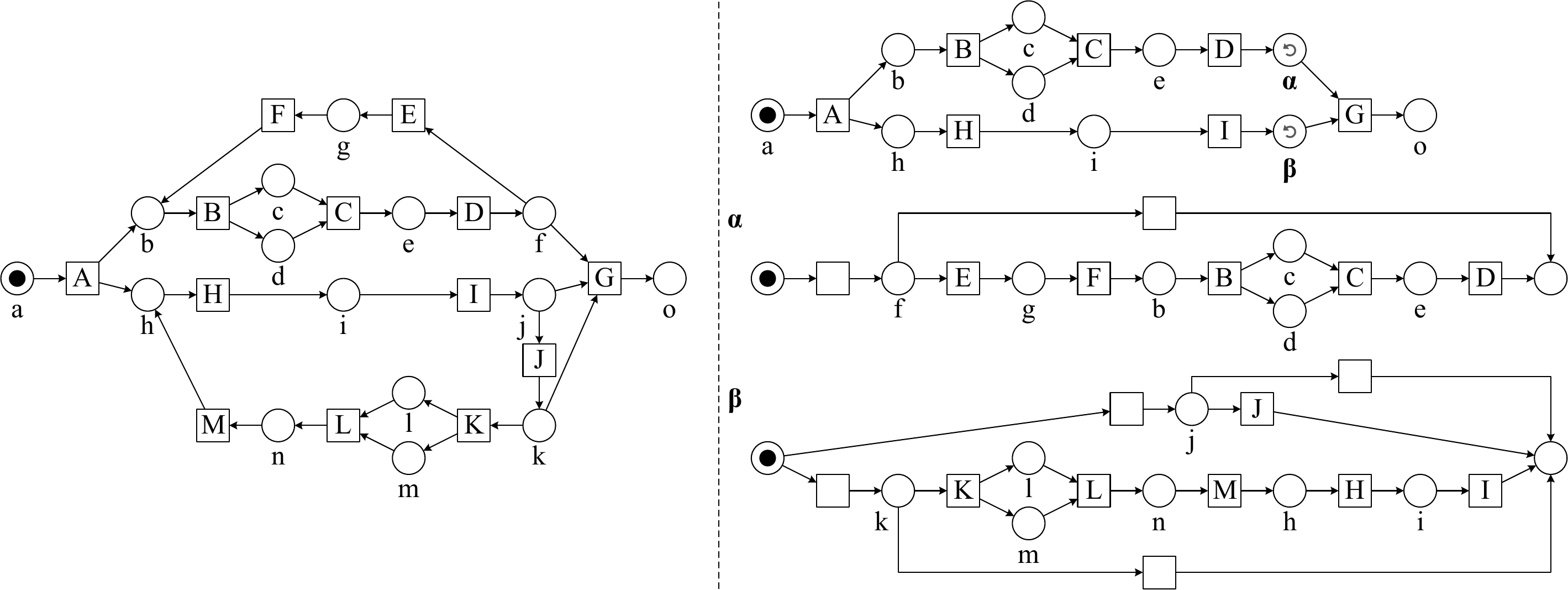}
	\caption{A sound cyclic FW-net with two loops (left) and its decomposition into three sound acyclic nets (right). The net $\alpha$ in the middle and the net $\beta$ in the bottom represent both loops, which were reduced in the net above with (loop) places $\alpha$ and $\beta$.}
	\label{fig:ExampleCyclicPN}
\end{figure}

\subsection{The Algorithm}

The overall situation after loop decomposition is a set of sound AFW-nets, in which some \emph{loop} places are linked to acyclic \emph{loop} nets. It is important to understand for correctness that the replacement of (parts of) loops with loop places does not change the concurrency behavior. That is, in sound FW-nets, an entire loop (after copying the do-body) acts as a place ``globally'' \cite{DBLP:conf/bpm/PrinzCH22}. For the computation of the concurrency relation it follows that each node in a concurrency relation with a \emph{loop place} is concurrent with each node of the linked loop. Newly inserted transitions and places in the loop nets to facilitate the decomposition are \emph{not} of interest, naturally. Nodes of loops in the do-body appear (at least) twice --- once in the surrounding net and once in the corresponding loop net. For instance, places $b$ and $e$ in Figure~\ref{fig:ExampleCyclicPN} are in the AFW-net in the top and in the $\alpha$-loop net in the middle. This does not have an influence on the final result, since if a node remaining in the surrounding net is in concurrency relation with another node, both nodes also are in concurrency relation in the original net. The same appears with nodes represented by the loop place being in concurrency relation with other nodes.

Algorithm~\ref{alg:ConcurrencyCyclic} describes the \emph{Concurrent Paths} (\emph{CP}) algorithm. It computes the adjacency set $R$ of concurrent nodes. This is initialized with an empty set in line~2. Line~3 calls Algorithm~\ref{algo:LoopDecomposition} to decompose net $WN$ into a set of acyclic nets $AcyclicNets$ and a $Connections$ relation. Lines~4--5 compute (and extend) $R$ for each acyclic net $WN_a$ with Algorithm~\ref{alg:Concurrency}. Lines~6--16 replace concurrency relations between loop places and other nodes. In doing this, it considers each connection between a loop place $l$ and a linked loop net $A$ (as combination of places and transitions) in line~6. For $l$, it considers each node $c$ that is concurrent to $l$ (line~7). If $c$ is also a loop place with a linked loop net, the nodes of that loop net are assigned to $B$ (lines~8--9). Otherwise, if $c$ is an ordinary node, $B$ only consists of $\{ c \}$ (line~11). Furthermore, although $c$ is concurrent to the loop place $l$, it is not of interest since $l$ is ``virtual'' (line~12). Then, each node $a$ in $A$ is concurrent to each node in $B$ (lines~13--14). In addition, each node $b$ in $B$ is concurrent to each node in $A$ (lines~15--16). 

\begin{algorithm}[tb]
\caption{The \emph{Concurrent Paths} (\emph{CP}) algorithm: A cyclic version of the algorithm to determine $\parallel$ for a sound FW-net $WN=(P,T,F,i,o)$.}
\label{alg:ConcurrencyCyclic}
\begin{algorithmic}[1]
\Function{ConcurrentPaths}{$WN=(P,T,F,i,o)$}
	\State $R \gets \emptyset$
	\State $AcyclicNets$, $Connections$ $\gets$ \Call{decomposeLoops}{$WN$}
	\ForAll{$WN_{a} \in AcyclicNets$}
		\State $R \gets$ \Call{determineConcurrency}{$WN_{a}$, $R$}		
	\EndFor	
	\ForAll{$(l,A) \in Connections$}
		\ForAll{$c \in R(l)$} 
			\If{$c \in Connections$}
				\State $B \gets Connections(c) \cap (P \cup T)$				
			\Else			
				\State $B \gets \{ c \}$
			\EndIf
			\State $R(c) \gets R(c) \setminus \{ l \}$
			\ForAll{$a \in A$}
				\State $R(a) \gets R(a) \cup B$
			\EndFor
			\ForAll{$b \in B$}
				\State $R(b) \gets R(b) \cup A$
			\EndFor
		\EndFor
	\EndFor	
	\State \Return $R$
\EndFunction
\end{algorithmic}
\end{algorithm}

The time complexity of the algorithm depends on the time complexity of loop decomposition as well as on the (possible) increase of the problem size. Loop decomposition is achievable in $O(|P|^2 + |P| \cdot |T| + |P| \cdot |F|)$ \cite{DBLP:conf/bpm/PrinzCH22,Prinz2023Preprint}. In the worst case of many nested loops, the problem size after loop decomposition may increase quadratically, i.\,e., instead of checking one net of size $|P| + |T| + |F|$, $|P|$ nets of size $|P| + |T| + |F|$ must be investigated as a result of the decomposition (in general, the number of (nested) loops in a net is limited by a low constant). Thus, the complexity of checking a single cyclic net may increase to $O(|P|^2 + |T| \cdot |P| + |P| \cdot |F|)$ in the worst case. In this worst case, the application of Algorithm~\ref{alg:Concurrency} on $|P|$ nets finally leads to a $O(|P|^3 + |P| \cdot |T|^2)$ and, therefore, cubic time complexity. For this reason, Algorithm~\ref{alg:ConcurrencyCyclic} has the same time behavior as the \emph{KovEs} algorithm in the worst case. Although it seems that Algorithm~\ref{alg:ConcurrencyCyclic} has no benefit regarding the \emph{KovEs} algorithm, there are several reasons why there are situations, in which the algorithm has its advantages: %
\begin{enumerate}[label=(\roman*)]
	\item The algorithm can be parallelized. At first, each resulting acyclic net can be analyzed in parallel. The combination of the results is possible in at most quadratic time. At second, once $HasPath$ is computed for a net in at most quadratic time, all transitions can be computed in parallel. The combination of the results is at most quadratic as well. 
	\item It is relatively rare that nets have many (nested) loops \cite{DBLP:conf/bpm/PrinzCH22}. In addition, if a loop has only one entry, it is not necessary to consider the do-body. The same holds true if the do-body does not contain any converging transition.
	\item If the net has a high degree of concurrency involving many nodes, the algorithm should have a performance benefit compared to \emph{KovEs}.
	\item Of course, if the net is acyclic, the algorithm is faster.
\end{enumerate}

\subsection{Example}

To illustrate how Algorithm~\ref{alg:ConcurrencyCyclic} works, we use the cyclic example net of Figure~\ref{fig:ExampleCyclicPN}. At first, loop decomposition decomposes the cyclic net into the acyclic nets on the right side of Figure~\ref{fig:ExampleCyclicPN}. Subsequently, it performs for each of the acyclic nets the acyclic \emph{CP} Algorithm~\ref{alg:Concurrency}. Table~\ref{tab:ExHasPathR} contains the relation $HasPath$ and the temporary relation $R$ for each of the acyclic nets and each of their nodes. For example, if transition $A$ is investigated, the pair $(b,h)$ is considered. The sets $HasPath(b)$ and $HasPath(h)$ can be found in Algorithm~\ref{tab:ExHasPathR}. It is valid that $HasPath(b) \setminus HasPath(h) = \{ b, B, c, d, C, e, D, \alpha \}$. For each of these nodes, we can extend the relation to, e.\,g., $R(b) = R(b) \cup \big( HasPath(b) \setminus HasPath(h) \big) = R(b) \cup \{ h, H, i, I, \beta \}$.

\begin{table}[tb]
	\centering
	\showpicture{\includegraphics[width=0.4\textwidth]{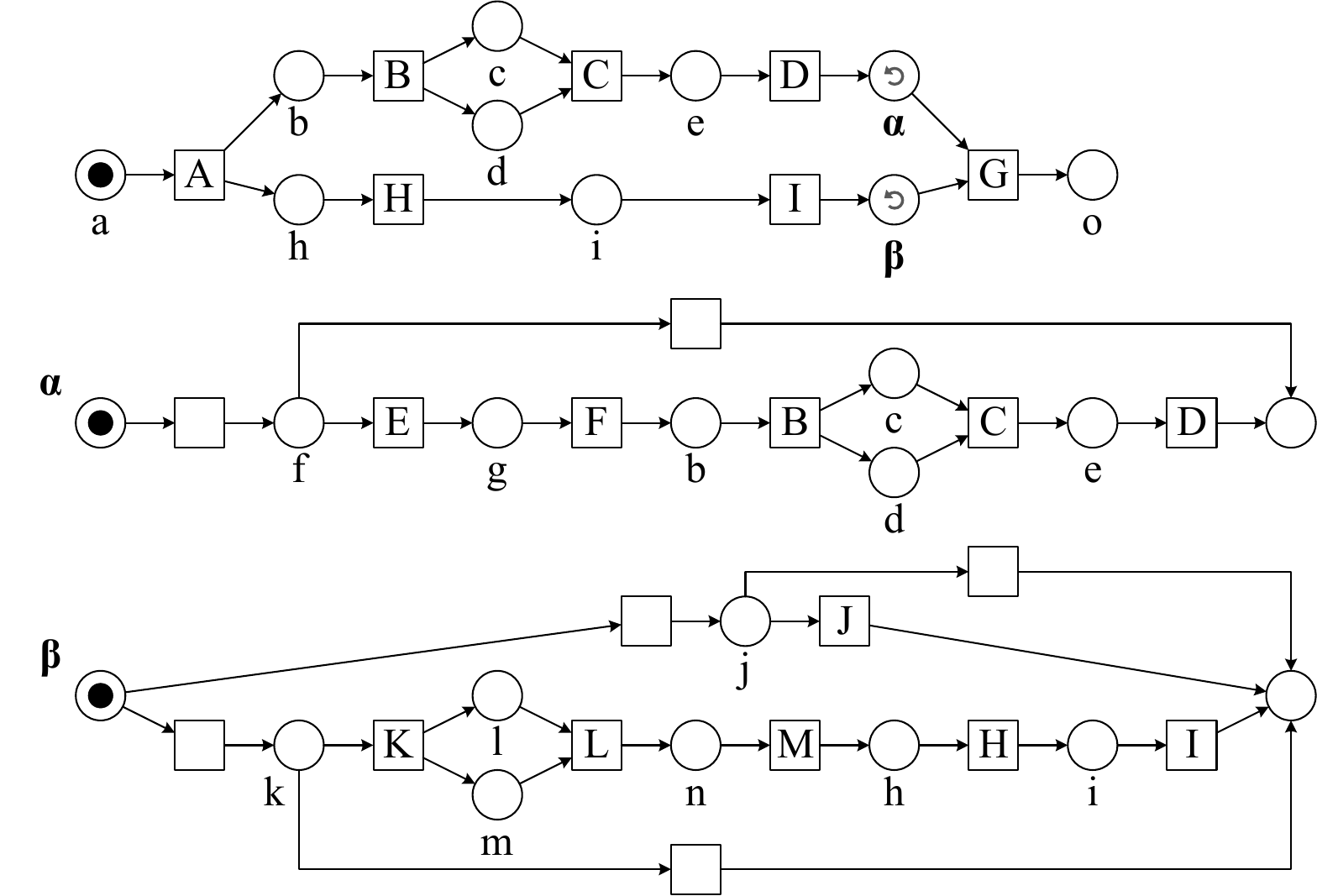} \\ \vspace{0.5cm}}
	\scriptsize
		\begin{tabular}{ccccp{5.6cm}cp{2.6cm}}
			Net && Node(s) && $HasPath$ && $R$ ($\parallel$) \\
			\cline{1-1} \cline{3-3} \cline{5-5} \cline{7-7} \\
			\multirow{\netonesize}{*}{1} \extakeout{&& $a$ && %
			$a$, $A$, $b$, $B$, $c$, $d$, $C$, $e$, $D$, $\alpha$, $G$, $h$, $H$, $i$, $I$, $\beta$, $o$ && %
			$\emptyset$ \\
			&& $A$ && %
			$A$, $b$, $B$, $c$, $d$, $C$, $e$, $D$, $\alpha$, $G$, $h$, $H$, $i$, $I$, $\beta$, $o$ && %
			$\emptyset$ \\}{&& $a$, $A$, $G$, $o$ && \emph{[Not of interest for this example]} && $\emptyset$ \\}
			&& $b$ && %
			$b$, $B$, $c$, $d$, $C$, $e$, $D$, $\alpha$, $G$, $o$ && %
			$h$, $H$, $i$, $I$, $\beta$ \\
			&& $B$ && %
			$B$, $c$, $d$, $C$, $e$, $D$, $\alpha$, $G$, $o$ && %
			$h$, $H$, $i$, $I$, $\beta$ \\
			&& $c$ && %
			$c$, $C$, $e$, $D$, $\alpha$, $G$, $o$ && %
			$d$, $h$, $H$, $i$, $I$, $\beta$ \\
			&& $d$ && %
			$d$, $C$, $e$, $D$, $\alpha$, $G$, $o$ && %
			$c$, $h$, $H$, $i$, $I$, $\beta$ \\
			&& $C$ && %
			$C$, $e$, $D$, $\alpha$, $G$, $o$ && %
			$h$, $H$, $i$, $I$, $\beta$ \\
			&& $e$ && %
			$e$, $D$, $\alpha$, $G$, $o$ && %
			$h$, $H$, $i$, $I$, $\beta$ \\
			&& $D$ && %
			$D$, $\alpha$, $G$, $o$ && %
			$h$, $H$, $i$, $I$, $\beta$ \\
			&& $\alpha$ && %
			$\alpha$, $G$, $o$ && %
			$h$, $H$, $i$, $I$, $\beta$ \\			
			&& $h$ && %
			$h$, $H$, $i$, $I$, $\beta$, $G$, $o$ && %
			$b$, $B$, $c$, $d$, $C$, $e$, $D$, $\alpha$ \\
			&& $H$ && %
			$H$, $i$, $I$, $\beta$, $G$, $o$ && %
			$b$, $B$, $c$, $d$, $C$, $e$, $D$, $\alpha$ \\
			&& $i$ && %
			$i$, $I$, $\beta$, $G$, $o$ && %
			$b$, $B$, $c$, $d$, $C$, $e$, $D$, $\alpha$ \\
			&& $I$ && %
			$I$, $\beta$, $G$, $o$ && %
			$b$, $B$, $c$, $d$, $C$, $e$, $D$, $\alpha$ \\
			&& $\beta$ && %
			$\beta$, $G$, $o$ && %
			$b$, $B$, $c$, $d$, $C$, $e$, $D$, $\alpha$ \\
			\extakeout{&& $G$ && %
			$G$, $o$ && %
			$\emptyset$ \\}{}
			\extakeout{&& $o$ && %
			$o$ && $\emptyset$ \\}{}
			\cline{1-1} \cline{3-3} \cline{5-5} \cline{7-7} \\
			\multirow{\nettwosize}{*}{2} \extakeout{&& $f$ && %
			$f$, $E$, $g$, $F$, $b$, $B$, $c$, $d$, $C$, $e$, $D$ && %
			$\emptyset$ \\
			&& $E$ && %
			$E$, $g$, $F$, $b$, $B$, $c$, $d$, $C$, $e$, $D$ && %
			$\emptyset$ \\
			&& $g$ && %
			$g$, $F$, $b$, $B$, $c$, $d$, $C$, $e$, $D$ && %
			$\emptyset$ \\
			&& $F$ && %
			$F$, $b$, $B$, $c$, $d$, $C$, $e$, $D$ && %
			$\emptyset$ \\
			&& $b$ && %
			$b$, $B$, $c$, $d$, $C$, $e$, $D$ && %
			$\emptyset$ \\
			&& $B$ && %
			$B$, $c$, $d$, $C$, $e$, $D$ && %
			$\emptyset$ \\}{&& $f$, $E$, $g$, $F$, $b$, $B$, $c$, $d$, $C$, $e$, $D$ && \emph{[Not of interest for this example]} && $\emptyset$ \\}
			&& $c$ && %
			$c$, $C$, $e$, $D$ && %
			$d$ \\
			&& $d$ && %
			$d$, $C$, $e$, $D$ && %
			$c$ \\
			\extakeout{&& $C$ && %
			$C$, $e$, $D$ && %
			$\emptyset$ \\
			&& $e$ && %
			$e$, $D$ && %
			$\emptyset$ \\
			&& $D$ && %
			$D$ && %
			$\emptyset$ \\}{}
			\cline{1-1} \cline{3-3} \cline{5-5} \cline{7-7} \\
			\multirow{\netthreesize}{*}{3} \extakeout{&& $j$ && %
			$j$, $J$ && %
			$\emptyset$ \\
			&& $J$ && %
			$J$ && %
			$\emptyset$ \\
			&& $k$ && %
			$k$, $K$, $l$, $m$, $L$, $n$, $M$, $h$, $H$, $i$, $I$ && %
			$\emptyset$ \\
			&& $K$ && %
			$K$, $l$, $m$, $L$, $n$, $M$, $h$, $H$, $i$, $I$ && %
			$\emptyset$ \\}{&& $j$, $J$, $k$, $K$, $L$, $n$, $M$, $h$, $H$, $i$, $I$ && \emph{[Not of interest for this example]} && $\emptyset$ \\}
			&& $l$ && %
			$l$, $L$, $n$, $M$, $h$, $H$, $i$, $I$ && %
			$m$ \\
			&& $m$ && %
			$m$, $L$, $n$, $M$, $h$, $H$, $i$, $I$ && %
			$l$ \\
			\extakeout{&& $L$ && %
			$L$, $n$, $M$, $h$, $H$, $i$, $I$ && %
			$\emptyset$ \\
			&& $n$ && %
			$n$, $M$, $h$, $H$, $i$, $I$ && %
			$\emptyset$ \\
			&& $M$ && %
			$M$, $h$, $H$, $i$, $I$ && %
			$\emptyset$ \\
			&& $h$ && %
			$h$, $H$, $i$, $I$ && %
			$\emptyset$ \\
			&& $H$ && %
			$H$, $i$, $I$ && %
			$\emptyset$ \\
			&& $i$ && %
			$i$, $I$ && %
			$\emptyset$ \\
			&& $I$ && %
			$I$ && %
			$\emptyset$ \\}{}
		\end{tabular}
	\caption{The table shows for each net and their nodes being illustrated in the top, the relations $HasPath$ and $R$. For $HasPath$, we omitted the mention for nodes, which are not in concurrency relation with another node. Both relations were derived after performing Algorithm~\ref{alg:Concurrency}.}
	\label{tab:ExHasPathR}
\end{table}

Once all acyclic nets are investigated, the results can be combined. For the upper acyclic net in Figure~\ref{fig:ExampleCyclicPN}, there are two loop places $\alpha$ and $\beta$. Let us take $\alpha$ as an example: $\alpha$ is already concurrent to $h$, $H$, $i$, $I$, and $\beta$ (cf. Algorithm~\ref{tab:ExHasPathR}). Each of these nodes is considered by Algorithm~\ref{alg:ConcurrencyCyclic}, lines~7--16. It may start with node $h$, which is not a loop place but an ordinary node. Therefore, $B \gets \{h\}$ and $h$ is not concurrent to $\alpha$ anymore (lines~10--12). However, each node in the corresponding net of $\alpha$ (the middle net in Figure~\ref{fig:ExampleCyclicPN}), is concurrent to $h$ (lines~13--14); and $h$ (in $B$) is concurrent to each node of the middle acyclic net. The same holds true for the case of $\beta$: It is a loop place so that $B$ contains all nodes of the corresponding net to $\beta$ (the lower net of Figure~\ref{fig:ExampleCyclicPN}), lines~8--9. That means, each node of the middle net is concurrent to each node of the lower net and, naturally, vice versa (lines~13--16). Table~\ref{tab:ExFinal} summarizes the $R$ ($\parallel$) relations for each node of the original net of Figure~\ref{fig:ExampleCyclicPN}.

\begin{table}[tb]
	\centering
	\showpicture{\includegraphics[width=0.4\textwidth]{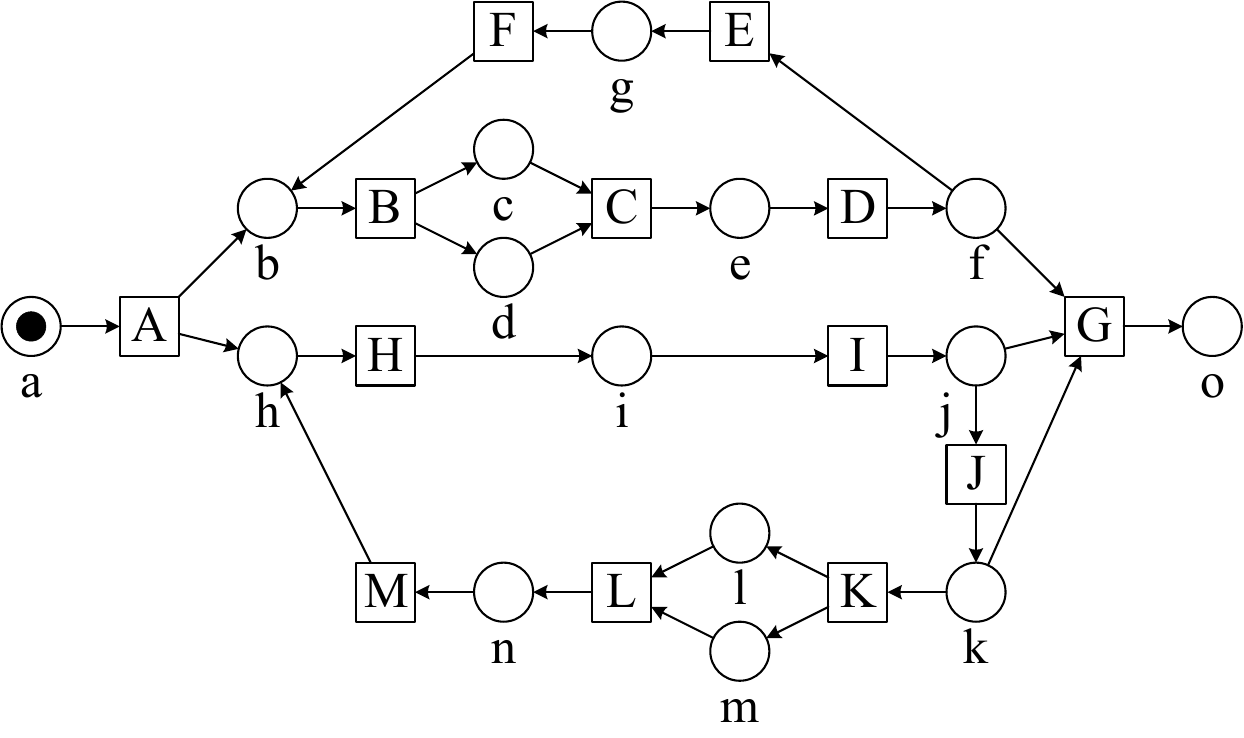} \\ \vspace{0.5cm}}
	\scriptsize
		\begin{tabular}{ccp{4.7cm}cccp{4.1cm}}
			Node && $R$ ($\parallel$) && Node && $R$ ($\parallel$) \\
			\cline{1-1} \cline{3-3} \cline{5-5} \cline{7-7} \\
			$a$ && %
			$\emptyset$ && %
			$h$ && %
			$b$, $B$, $c$, $d$, $C$, $e$, $D$, $f$, $E$, $g$, $F$ \\
			$A$ && %
			$\emptyset$ && %
			$H$ && %
			$b$, $B$, $c$, $d$, $C$, $e$, $D$, $f$, $E$, $g$, $F$ \\
			$b$ && %
			$h$, $H$, $i$, $I$, $j$, $J$, $k$, $K$, $l$, $m$, $L$, $n$, $M$ && %
			$i$ && %
			$b$, $B$, $c$, $d$, $C$, $e$, $D$, $f$, $E$, $g$, $F$ \\
			$B$ && %
			$h$, $H$, $i$, $I$, $j$, $J$, $k$, $K$, $l$, $m$, $L$, $n$, $M$ && %
			$I$ && %
			$b$, $B$, $c$, $d$, $C$, $e$, $D$, $f$, $E$, $g$, $F$ \\
			$c$ && %
			$d$, $h$, $H$, $i$, $I$, $j$, $J$, $k$, $K$, $l$, $m$, $L$, $n$, $M$ && %
			$j$ && %
			$b$, $B$, $c$, $d$, $C$, $e$, $D$, $f$, $E$, $g$, $F$ \\
			$d$ && %
			$c$, $h$, $H$, $i$, $I$, $j$, $J$, $k$, $K$, $l$, $m$, $L$, $n$, $M$ && %
			$J$ && %
			$b$, $B$, $c$, $d$, $C$, $e$, $D$, $f$, $E$, $g$, $F$ \\
			$C$ && %
			$h$, $H$, $i$, $I$, $j$, $J$, $k$, $K$, $l$, $m$, $L$, $n$, $M$ && %
			$k$ && %
			$b$, $B$, $c$, $d$, $C$, $e$, $D$, $f$, $E$, $g$, $F$ \\
			$e$ && %
			$h$, $H$, $i$, $I$, $j$, $J$, $k$, $K$, $l$, $m$, $L$, $n$, $M$ && %
			$K$ && %
			$b$, $B$, $c$, $d$, $C$, $e$, $D$, $f$, $E$, $g$, $F$ \\
			$D$ && %
			$h$, $H$, $i$, $I$, $j$, $J$, $k$, $K$, $l$, $m$, $L$, $n$, $M$ && %
			$l$ && %
			$m$, $b$, $B$, $c$, $d$, $C$, $e$, $D$, $f$, $E$, $g$, $F$ \\
			$f$ && %
			$h$, $H$, $i$, $I$, $j$, $J$, $k$, $K$, $l$, $m$, $L$, $n$, $M$ && %
			$m$ && %
			$l$, $b$, $B$, $c$, $d$, $C$, $e$, $D$, $f$, $E$, $g$, $F$ \\
			$E$ && %
			$h$, $H$, $i$, $I$, $j$, $J$, $k$, $K$, $l$, $m$, $L$, $n$, $M$ && %
			$L$ && %
			$b$, $B$, $c$, $d$, $C$, $e$, $D$, $f$, $E$, $g$, $F$ \\
			$g$ && %
			$h$, $H$, $i$, $I$, $j$, $J$, $k$, $K$, $l$, $m$, $L$, $n$, $M$ && %
			$n$ && %
			$b$, $B$, $c$, $d$, $C$, $e$, $D$, $f$, $E$, $g$, $F$ \\
			$F$ && %
			$h$, $H$, $i$, $I$, $j$, $J$, $k$, $K$, $l$, $m$, $L$, $n$, $M$ && %
			$M$ && %
			$b$, $B$, $c$, $d$, $C$, $e$, $D$, $f$, $E$, $g$, $F$ \\
			$G$ && %
			$\emptyset$ && %
			$o$ && %
			$\emptyset$ \\
		\end{tabular}
	\caption{The nodes of the net in the top being in a concurrency relation $\parallel$ after performing Algorithm~\ref{alg:ConcurrencyCyclic} and combining the $R$ relations in Algorithm~\ref{tab:ExFinal}.}
	\label{tab:ExFinal}
\end{table}


\section{Inclusive Semantics in Process Models}
\label{sec:InclusiveSemantics}

Inclusive semantics is an in-between between exclusive semantics (being achievable with places with at least two transitions in their postsets) and concurrent semantics (being represented with transitions with at least two places in their postset). In inclusive semantics, a non-empty subset of nodes in an inclusive node's postset get tokens and not each place in a node's preset must carry a token to be enabled. Such a semantics is not directly representable in Petri nets, however, an usual element of business process models in business process management. Since process models are usually represented as Petri nets or have at least borrowed semantics, handling inclusive semantics in concurrency detection is of benefit. Although diverging inclusive nodes (``OR-splits'') can be represented with variations of places and transitions in nets; representing converging inclusive nodes (``OR-joins'') may lead to free-choice nets being grown exponentially \cite{DBLP:journals/is/FavreFV15}. Since process models are transformed into nets for analysis in many cases, the non-free-choice property or its exponential growth may lead to computational problems.

The \emph{KovEs} algorithm is \emph{not} applicable to investigate inclusive semantics in process models without strong modifications. However, the \emph{CP} approach is directly applicable if the inclusive semantics proposed in previous work \cite{DBLP:journals/csimq/PrinzA15} and resulting from loop decomposition \cite{DBLP:conf/bpm/PrinzCH22} is assumed.


\section{Evaluation}
\label{sec:Evaluation}

The \emph{KovEs} algorithm and the presented \emph{CP} algorithm in Section~\ref{sec:CyclicNets} have been implemented in a simple script-based algorithm (PHP) for the purpose of evaluation. The implementation is open-source and available on GitHub\footnote{\url{https://github.com/guybrushPrince/cp}}. The following experiments were conducted on a machine with an Intel{\textregistered} Core{\texttrademark} i7 CPU with 4 cores, 16 GB of main memory, and Microsoft Windows 11 Professional. PHP was used in version 8. We performed all measures for 10 times, removed the best and worst times, and used the mean values of all the remaining runs.

Both algorithms were applied to a well-known dataset, namely the \emph{IBM Websphere Business Modeler} dataset \cite{DBLP:journals/dke/FahlandFKLVW11}, which consists of 1,368 files and is referred to as \emph{IBM} hereafter. Only 644 nets of the IBM dataset can be investigated since the algorithm requires sound FW-nets. The nets were available as Petri Net Markup Language (PNML) models. The comparison of both algorithms was restricted to the identification of concurrent places instead of concurrent places and transitions, since the \emph{KovEs} algorithm cannot find concurrent transitions without modifications. The nets under investigation are small (75\% with less or equal 58 nodes) to big (with a maximum of 546 nodes). Places are more frequent than transitions (approx. 61\%$\pm$4\% of all nodes are places). With regard to the number of nodes, a net has approx. 111\%$\pm$11.5\% flows.

Both algorithms find the same places being concurrent for all suitable nets of the IBM dataset. Although the total number of nodes $|N|$ in a net bounds the number of concurrent nodes $|\parallel| \leq |(P+T)^2|$, the number of nodes ($R^2 = 0.2$), places ($R^2 = 0.42$), and transitions ($R^2 = 0.06$) do not well explain the degree of concurrency in a net by applying a linear regression considering whether $|\parallel| \sim |(P+T)^2|$, $|\parallel| \sim |P^2|$, or $|\parallel| \sim |T^2|$, respectively.

The overall goal was to construct a more efficient algorithm than the \emph{KovEs} algorithm. For this reason, we have compared the times both algorithms need to compute all concurrency relations for a net. The \emph{CP} algorithm was overall faster for the entire dataset. It just needs 285 [ms] to compute 192,170 pairs of places being in relation. On contrary, the \emph{KovEs} algorithm requires 14,100 [ms] for doing the same job, i.\,e., the \emph{CP} algorithm is approx. 50 times faster. On closer inspection, \emph{KovEs} has its benefits for nets without a high degree of concurrency. For these cases, it performs better than \emph{CP}. Figure~\ref{fig:VisNodesNodes} shows a chart comparing the number of investigated nodes during both algorithms in relation to the number of nodes in the net. The y-axis is scaled logarithmically. The \emph{CP} algorithm has, of course, a higher ``start up'' number of nodes to investigate, because it visits at least each node and flow once for the computation of paths. \emph{KovEs} instead directly starts with the computation of the concurrency relations and, therefore, has no start up number of nodes and has a better performance for nets with less concurrent nodes. Figure~\ref{fig:VisNodesPairs} illustrates the number of investigated nodes in relation to the number of relations. The chart reveals that the computational load of \emph{KovEs} seems to increase slightly quadratically if the number of nodes in relation increase. Instead, the load of \emph{CP} seems to increase only linearly for a higher number of relations. This correlation between the computational load and the number of relations becomes more obvious if we compare the computation time in relation to the number of nodes being concurrent as it is done in Figure~\ref{fig:TimePairs}. The computation time of the \emph{KovEs} algorithm seems to increase quadratically to the number of concurrent nodes; instead, the computation time of \emph{CP} seems to increase just linearly.

\begin{figure}[tb]
	\centering
		\includegraphics[width=0.70\textwidth]{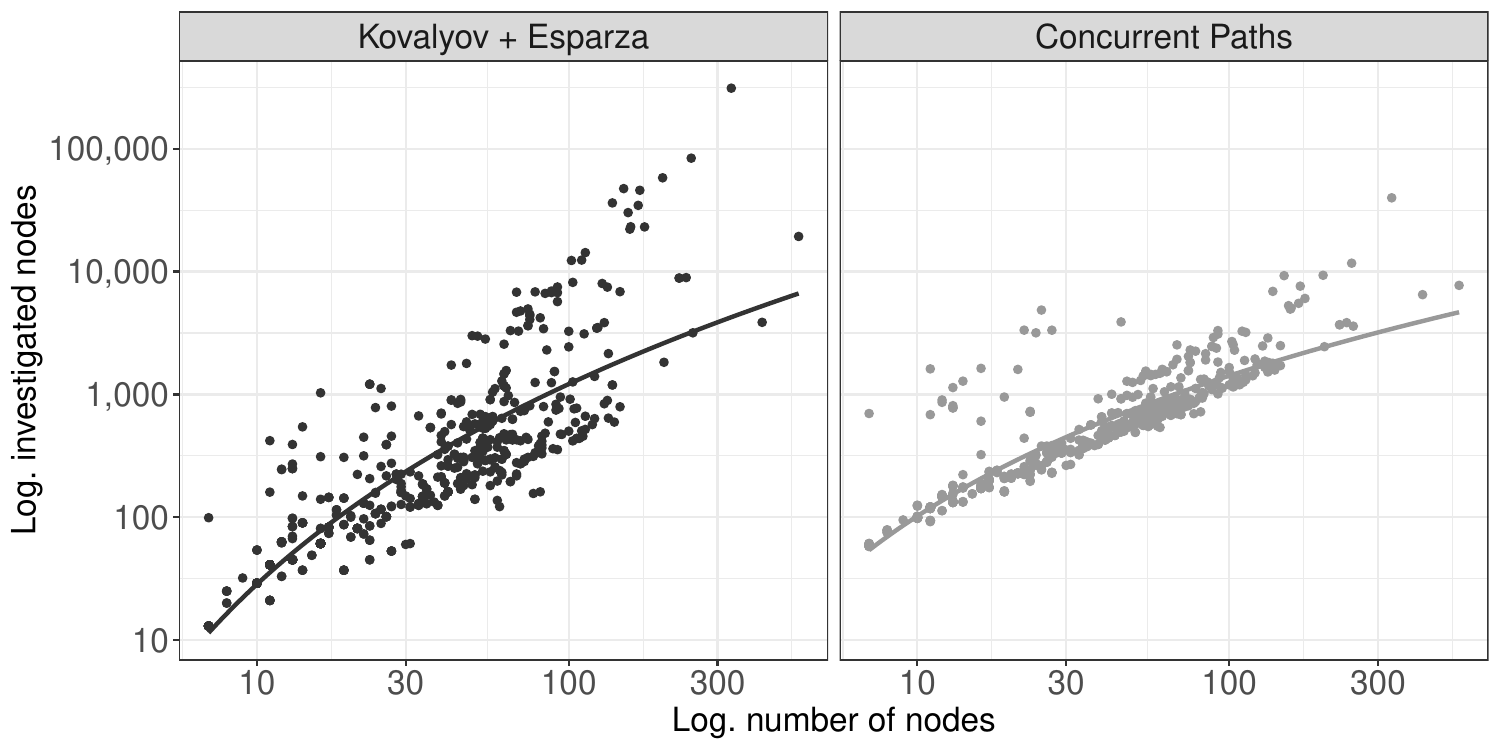}
	\caption{The number of investigated nodes during computation compared to the number of nodes in the nets.}
	\label{fig:VisNodesNodes}
\end{figure}

\begin{figure}[tb]
	\centering
	  \hfill 
		\begin{minipage}[t]{0.45\textwidth}
			\includegraphics[width=1.0\textwidth]{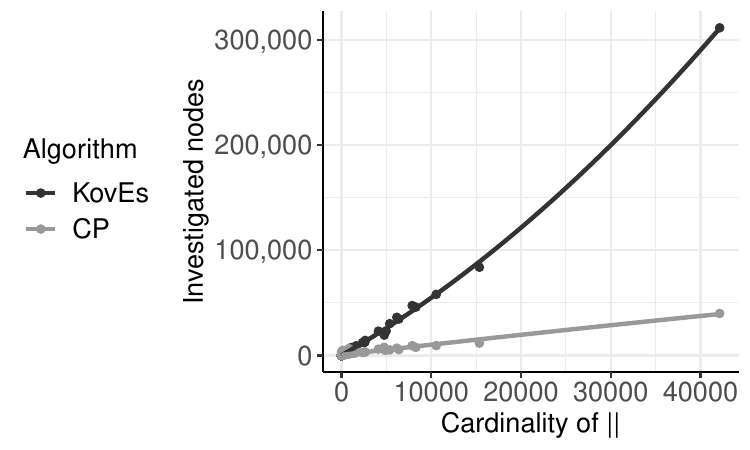}
		\end{minipage} \begin{minipage}[t]{0.45\textwidth}
			\includegraphics[width=1.0\textwidth]{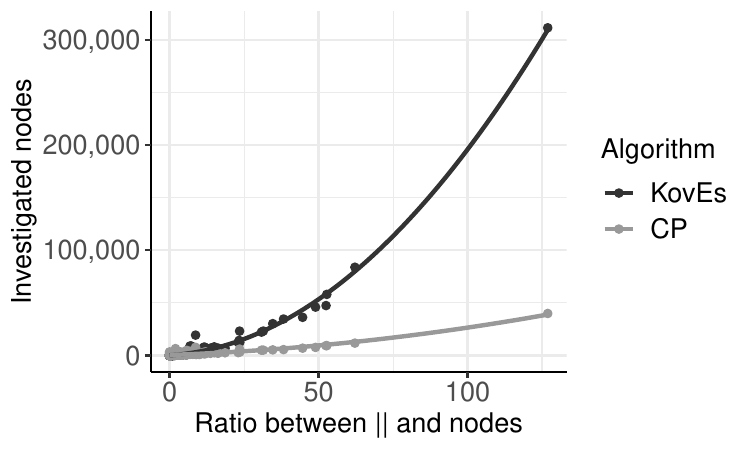}
		\end{minipage}
		\hfill 
	\caption{The number of investigated nodes during computation compared to the number of pairs in concurrency relation (left) and the ratio between the number of pairs and the number of nodes (right).}
	\label{fig:VisNodesPairs}
\end{figure}

\begin{figure}[tb]
	\centering
	\includegraphics[width=0.8\textwidth]{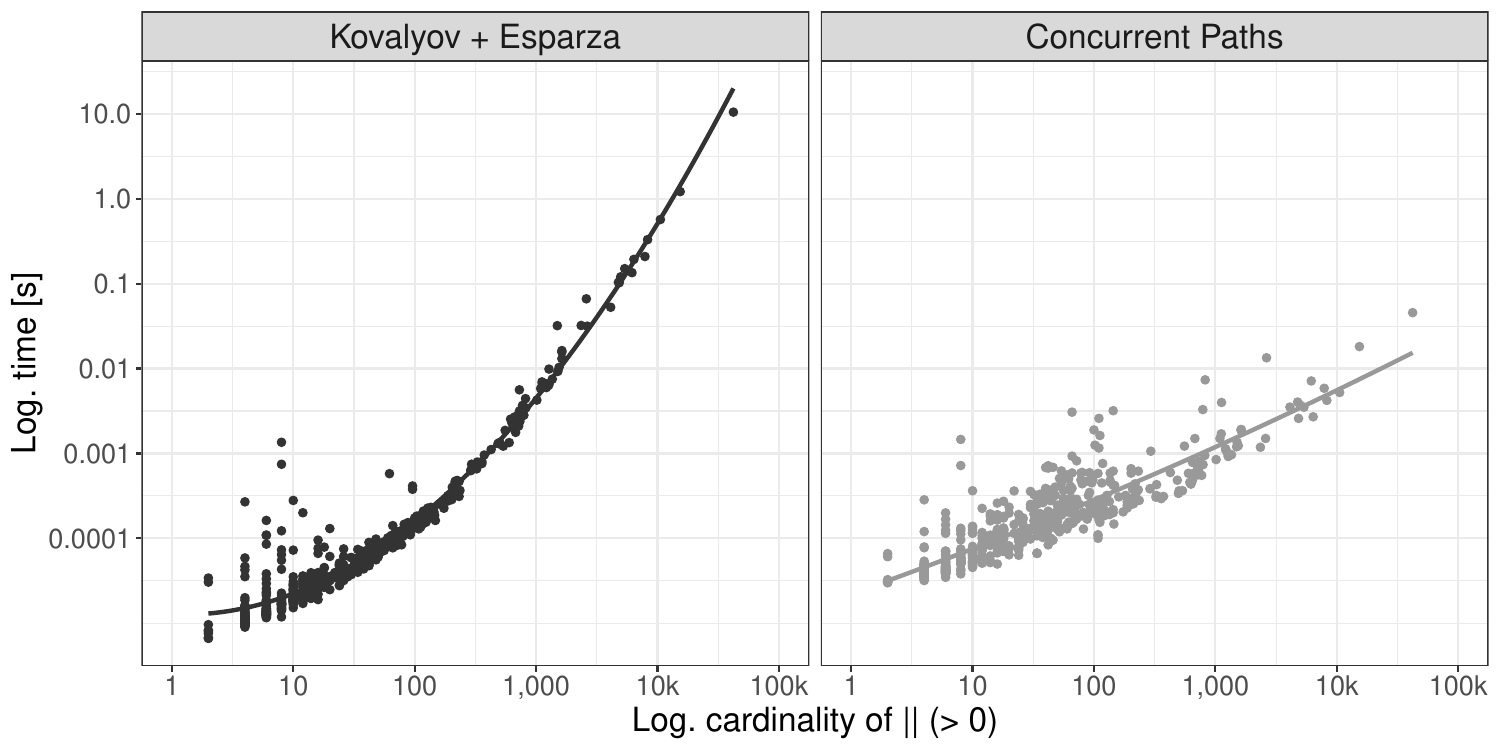}	
	\caption{Computation time regarding the number of pairs of nodes in concurrency relation.}
	\label{fig:TimePairs}
\end{figure}

Only 40 out of 644 nets of the IBM dataset are cyclic. The acyclic nets are investigated in 241 [ms] with \emph{CP} and in 14,084 [ms] with \emph{KovEs}. In fact, one net with approx. 42k concurrent pairs alone needs 10,500 [ms] for \emph{KovEs} but just 45 [ms] with \emph{CP}. \emph{CP} is faster than \emph{KovEs} for 87 nets. In sum, \emph{CP} needs 171 [ms] and \emph{KovEs} needs 14,06 [ms] for these 87 nets (speedup factor approx. 82), whereas for the other 557 nets, \emph{CP} needs 114 [ms] and \emph{KovEs} needs 41 [ms] (slowdown factor approx. 3). The 87 nets comprise approx. 172k of pairs in relation against approx. 20k for the others. Therefore, the 87 nets have a much stronger computational intensity. \emph{KovEs} is faster for the cyclic nets with 16 [ms] against 43 [ms] for \emph{CP}. This is reasonable because of three reasons: (1) \emph{CP} has a higher start up time especially for cyclic nets by performing a quadratic loop decomposition, (2) \emph{CP} has the same worst-case cubic runtime complexity like \emph{KovEs}, and (3) the ratio of concurrent nodes to the number of nodes is small for the 40 nets; there are just approx. 2.5 times of nodes in relation compared to the total number of nodes (i.\,e., a net with 100 nodes has around 250 pairs of nodes being concurrent). If this ratio is higher, \emph{CP} benefits against \emph{KovEs}. For showcasing this, the acyclic net with approx. 42k pairs of nodes in relation was surrounded with a simple loop. Although \emph{CP} has to perform loop decomposition, the computation times are similar to those of the acyclic case --- it shows that \emph{CP} is more efficient than \emph{KovEs} if a net contains many concurrent pairs of nodes.


\section{Conclusion}
\label{sec:Conclusion}

Concurrency detection identifies pairs of nodes that may be executed in parallel. Knowing concurrent places and transitions in Petri nets is essential to understand their behavior and is crucial as the base for ongoing analysis. Since such nets are potentially large and complex with many pairs of concurrent nodes, efficient algorithms are necessary. This paper extends the palette of concurrency detection algorithms with the \emph{Concurrent Paths} (CP) algorithm for sound free-choice workflow nets. For acyclic nets, the algorithm performs in quadratic time $O(P^2 + T^2)$ with $P$ the number of places and $T$ the number of transitions of a net. The algorithm requires a cubic time complexity in the worst-case for cyclic nets, $O(P^3 + PT^2)$. Although this seems not to be an improvement of the algorithm of Kovalyov and Esparza (\emph{KovEs}) (which needs a cubic time complexity for live and bounded free-choice nets), parallelizing \emph{CP} is straight-forward and the worst-case of \emph{CP} appears significantly less frequent than the worst-case of \emph{KovEs} and can be assumed to have just a constant impact on computation time for most nets. An evaluation of \emph{CP} on a benchmark of nets showed strong benefits on nets with a high degree of concurrency and only small disadvantages on nets with a low degree of concurrency. 

This paper enables Petri net analysis to be more efficient, especially, in cases of a high degree in concurrency of nets. Although single nets could be analyzed efficiently with \emph{KovEs}, performing concurrency detection on a large set of nets (e.g., for indexing in a database) may require much time. \emph{CP} reduces the effort and enables more efficient strategies to compute other properties of nets such as causality, exclusivity, etc. As a side effect, related research areas such as business process management and information systems research profit from the new technique as they utilize nets for analysis.

For future work, we plan to apply \emph{CP} to efficiently derive all relations in the 4C spectrum \cite{DBLP:conf/apn/PolyvyanyyWCRH14} for sound free-choice workflow nets. This would allow for an efficient indexing of nets in querying languages for nets, which is especially beneficial in business process management.

\bibliographystyle{acm}
\bibliography{arxivbibliography}

\end{document}